\documentclass[12pt]{article}
\usepackage{amsgen,amsmath,amstext,amsbsy,amsopn,amssymb,bbm, amsthm}
\usepackage{comment}
\usepackage[pdftex]{hyperref}
\usepackage{graphicx}
\usepackage{caption}
\usepackage{subcaption}
\usepackage{natbib}
\usepackage{fullpage}
\usepackage{authblk}
\usepackage{float}
\usepackage{xcolor}
\usepackage{setspace}
\usepackage{parskip}

\usepackage{tikz}
\usetikzlibrary{calc}
\usetikzlibrary{bayesnet}

\newcommand{\blind}{1} 

\newtheorem{prop}{Property}

\newcommand{\E}{\mathbb{E}}
\newcommand{\I}{\mathbf{I}}

\newcommand{\bt}{\mathbf{t}}
\newcommand{\bk}{\mathbf{k}}
\newcommand{\br}{\mathbf{r}}
\newcommand{\y}{\mathbf{y}}
\newcommand{\x}{\mathbf{x}}

\newcommand{\w}{\mathbf{w}}

\newcommand{\Btheta}{\boldsymbol{\theta}}

\newcommand{\SMproofs}{1}
\newcommand{\SMsplitmerge}{2}
\newcommand{\SMsmnHDP}{2.2}
\newcommand{\SMsimulations}{3}
\newcommand{\SMsimmixt}{3.1}
\newcommand{\SMsimmodel}{3.2}
\newcommand{\SMphilly}{4}
\newcommand{\SMfigure}{5}

\setcitestyle{semicolon}

\date{}

\begin{document}

\title{\bf Clustering Areal Units at Multiple Levels of Resolution to Model Crime in Philadelphia}
  \author[$\dag$]{Cecilia Balocchi}
\author[$\star$]{Edward I. George}
\author[$\star$]{Shane T. Jensen}
\affil[$\dag$]{{\small School of Mathematics, University of Edinburgh}}
\affil[$\star$]{{\small Department of Statistics, University of Pennsylvania}}
  \maketitle
  \vspace{-1.2cm}

\bigskip
\begin{abstract}
Estimation of the spatial heterogeneity in crime incidence across an entire city is an important step towards reducing crime and increasing our understanding of the physical and social functioning of urban environments.  This is a difficult modeling endeavor since crime incidence can vary smoothly across space and time but there also exist physical and social barriers that result in discontinuities in crime rates between different regions within a city.  A further difficulty is that there are different levels of resolution that can be used for defining regions of a city in order to analyze crime.  To address these challenges, we develop a Bayesian non-parametric approach for the clustering of urban areal units at different levels of resolution simultaneously.  Our approach is evaluated with an extensive synthetic data study and then applied to the estimation of crime incidence at various levels of resolution in the city of Philadelphia. 
\end{abstract}

\noindent%
{\it Keywords:}  Bayesian Hierarchical Modeling; Spatial Clustering; Hierarchical Dirichlet Process; Nested Dirichlet Process
\vfill

\onehalfspacing
\newpage

\section{Introduction}

The public availability of geocoded crime data provides the opportunity to analyze patterns in crime within urban environments at a higher resolution than ever before.  The focus on this paper will be to model the spatial heterogeneity in criminal activity in the city of Philadelphia, which is currently the sixth largest city in the United States.   Estimation of which local areas of a city have elevated levels of crime can potentially aid in efforts to reduce criminal activity as well as increase our scientific understanding about the physical and social functioning of cities.  Criminologists and urban planners are particularly interested in what aspects of the built environments are associated with elevated or reduced criminal activity.   Part of our investigation will examine our estimated crime patterns in the context of several historical theories and past studies from the urban planning and criminology literature.  

Modeling the spatial heterogeneity in crime across an entire city is a difficult endeavor as we expect that levels (and trends) in crime to vary somewhat smoothly across space and time in a city but we also must acknowledge that there exist both physical and social boundaries that may manifest as non-smooth breaks in crime patterns between proximal areas.  Adding to this challenge is that it is not clear what spatial resolution should be used when analyzing crime in urban environments, or even whether different spatial resolutions are appropriate for capturing the variation in crime in different regions of a large city.

The highest possible resolution for modeling crime is the geo-coded location of each individual reported crime, as the available crime data for Philadelphia gives the location of each reported crime as a GPS coordinate.  There have been many approaches to modeling geo-coded crime locations directly using spatial point processes, such as \cite{taddy2010autoregressive} and \cite{mohler2011self}.   However, a common alternative is to model crime as {\it areal data} where individual geo-coded crime locations are aggregated within a set of pre-defined geographic units \citep{aldor2016spatio,law2014bayesian, li2014space}.  

There are several compelling reasons for modeling crime occurrences within geographic {\it areal units}.  First,  policy initiatives and resources aimed towards reducing crime in cities are typically organized and enacted within areal units.   For example, the city of Philadelphia is organized into 21 police service areas (PSAs) and 66 police districts (PDs) by the Philadelphia Police Department for the purposes of administration and allocation of resources.   

Second, many important predictors or covariates of crime are only available as aggregate measures within areal units.  For example, socio-economic measures are predictive of crime in urban neighborhoods \citep{SamRauEar97}.   However, demographic and economic data is too sensitive to be made publicly available at the individual level and so the U.S. Census Bureau only makes available socio-economic measures aggregated at the level of their {\it census block} or {\it census block group} areal units.   In our crime analysis, we will focus on the 1336 U.S. Census block group areal units that are nested with 384 U.S. Census tract areal units within the city of Philadelphia.

Finally, modeling crime within areal units makes it easier to analyze the spatial heterogeneity in crime density in Philadelphia by providing an interpretation of these areal units as ``neighborhoods" of the city that each have unique features that could be associated with crime.  For example, \cite{humphrey2020urban} examined associations between crime and aspects of the built environment at the neighborhood level in Philadelphia, where neighborhoods were defined by the U.S. census block group areal units.   

In fact, there are many historical theories in criminology and urban planning about the connection between crime and different aspects of the surrounding neighborhood.  Several studies in criminology have found that crime tends to spatially concentrate at micro-levels \citep{weisburd2015law, lee2017concentrated}.   Defensible space theory \citep{New72}, crime prevention by environmental design \citep{CozSavHil05} and routine activities theory \citep{CohFel79} argue that the built environment in urban neighborhoods can impact propensities of offenders and availability of suitable targets for crime.  \cite{Mac15} and \cite{Macbrasto19} review many empirical studies have shown that crime incidence is associated with differences between local urban areas in terms of green space, zoning and public transit. 

However, when defining particular areal units as ``neighborhoods" for the purposes of analyzing crime in a city, it is important to acknowledge that crime levels may have more or less spatial variation, and that different levels of resolution may be more appropriate, in different parts of a large city.   Simultaneous modeling of crime between areal units across multiple levels of resolution is the focus of the methodological contributions of this paper. 

For areal units at a particular level of resolution, such as U.S. census block groups, we need a spatial modeling approach that allows for the sharing of information between areal units (neighborhoods) while also recognizing the possibility of sharp discontinuities in crime rates due to physical barriers such as rivers or social barriers such as differences in community organization between bordering neighborhoods.  

Modeling strategies based on the {\it clustering of areal units} can address these characteristics and are commonly used in spatial data situations such as disease mapping \citep{knorr2000bayesian,feng2016spatial,denison2001bayesian,anderson2017spatial} as well as crime \citep{Balocchi2019Clustering}.  

Clustering approaches also have advantages in terms of model interpretation, as they can reduce the dimensionality of the units of analysis from a potentially large number of individual areal units down to a smaller number of areal unit clusters.  This is an important feature for our application, as we want to identify regions of the city of Philadelphia that share similar crime dynamics in order to improve our understanding of what neighborhood factors are related to crime density.  

However, a crucial issue for investigating crime in Philadelphia (and spatial data modeling more generally) is choosing a particular {\it resolution} of areal units that should be the focus of our study or alternatively whether we can model crime density across {\it multiple resolutions} of areal units simultaneously.   

For example, as we mention above the city of Philadelphia is divided up into 21 police service areas and 66 police districts.  The U.S. Census Bureau provide even higher resolution options that could be used as areal units in our analysis: the city of Philadelphia is divided up into 384 census tracts which are further divided into 1336 census block groups.  Should we be modeling crime using the lower resolution census tracts or the higher resolution census block groups as our areal units? 

The vast majority of previous approaches to modeling crime, and spatial data more generally, fix a particular level of resolution for their areal units and only work with that choice of areal units for the entirety of their analysis \citep{feng2016spatial,anderson2017spatial,BalJen19}.  Some studies in criminology have investigated the best resolution for analyzing crime \citep{malleson2019identifying} whereas other studies have examined crime across different resolutions \citep{schnell2017influence,steenbeek2016action}.  

However, instead of fixing the level of resolution at which to perform our analysis, we want to model crime (and cluster areal units) in Philadelphia {\it simultaneously across multiple levels of resolution}.   We have already discussed several motivations for taking a clustering approach, but there are also several reasons for clustering areal units at different levels of resolution simultaneously.

One reason that is especially relevant to our study of crime is that one level of resolution may be most appropriate for certain areas of Philadelphia whereas another level of resolution may be more appropriate for other parts of the city.  For example, the spatial variation in crime density could be much higher in the more densely populated and heterogeneous central neighborhoods of Philadelphia compared with the more residential and homogeneous surrounding suburban neighborhoods.  If this is the case, we would ideally model crime density with a finer granularity of areal units in the central regions of the city compared to a coarser granularity of areal units in the outlying regions of the city.  This can be accomplished by simultaneously clustering areal units on multiple levels of resolution.   

This is a general issue beyond our application to crime in Philadelphia, as there are many spatial data situations in complex environments where one would expect that spatial variation is higher or lower in some regions under study compared to others, e.g. disease mapping, climate modeling, etc. 
Standard clustering methods that focus on a single level of granularity might not be able to capture those differences in variation.  Therefore, instead of restricting an analysis to a specific level of resolution, it can be beneficial to allow multiple levels of resolution that are each appropriate for subsets of our data.  

Another data situation where simultaneous modeling at multiple levels of resolution modeling can be beneficial is when part of the data, such as important predictor variables, are unavailable at the desired level of resolution.  Consider an analysis where we would like to cluster census blocks in a city based on crime density after accounting for measures of poverty via a regression model.  However, due to privacy issues, poverty data is only available from the US Census Bureau at the census block group level, which is a coarser resolution than the census block level.  In this case, the clustering of areal units at the desired higher resolution level of census blocks can be informed by a multi-level clustering model that also clusters areal units at the lower resolution level of census block groups at which we can condition on measures of poverty. 

Finally, a general goal of spatial modeling is capturing the dependency shared by areal units that are located in close geographical proximity to each other.  Parametric models can explicitly incorporate this dependency via spatial autocorrelation parameters, e.g. \cite{BalJen19}, though parametric approaches can have substantial modeling and computational challenges \citep{PacLeS10}.   The model for spatial clustering at multiple resolutions that we develop represents a less parametric alternative that will implicitly induce spatial dependence between geographically proximal areal units at a higher resolution level through the simultaneous clustering of the areal units at a lower resolution level.   

The multi-level modeling of areal units has been considered in various disciplines \citep{arcaya2012area,geronimus1998use,pickett2001multilevel,krieger2002geocoding}, but not specifically focussed on the problem of clustering areal units within a fixed hierarchical structure.  By fixed hierarchical structure, we are referring to our situation of having spatially proximal areal units at a higher resolution level (such as census block groups) that are nested within the areal units at the next lower resolution level (such as census tracts).   The approach of \citep{estivill2002multi} tries to find a hierarchical structure by finding clusters within clusters, but do not consider our situation of a fixed hierarchical structure of nested areal units across different levels of resolution.  

That said, neighborhoods that are not spatially proximal to each other in Philadelphia could still share similar levels of crime as well as socioeconomic situations. Moreover, they could have substantially different low-resolution crime density while still sharing similar high-resolution crime density, and so we do not want to restrict our clustering of areal units to only find partitions that are nested, i.e. high-resolution clusters nested within low-resolution clusters.  
In our modeling approach, we will need to acknowledge the fixed hierarchical structure of our areal units in the city of Philadelphia, but we also will need to build in the flexibility to potentially cluster high-resolution areal units together even if their corresponding low-resolution units are not in the same cluster. 

Before outlining our own methodological developments, we first review some recent modeling approaches that we will build upon to address our goal of simultaneously clustering areal units across multiple levels of resolution.  Bayesian non-parametric and semi-parametric models \citep{MulQui04} based on the Dirichlet process prior \citep{Fer73} have become an increasingly popular general approach for clustering observations and/or latent variables.  These models have the advantage over common approaches such as K-means clustering in that the number of clusters does not have to be fixed or pre-specified.  

\cite{teh2006hierarchical} developed a hierarchical Dirichlet Process model (HDP) for grouped data where the groups are known and fixed, whereas the nested Dirichlet process (nDP) of \citep{rodriguez2008nested} estimates the group structure from the data itself.   Thus, a nested Dirichlet Process model can be used for multi-level clustering of both observations and the groups containing those observations while the hierarchical Dirichlet Process model only produces a clustering of the observations.  

However, a key restriction of the nested Dirichlet Process is that clusters at finer levels are forced to be nested within the clusters at coarser levels, as shown in \cite{camerlenghi2018latent}.   As we discussed above, this characteristic of the nDP model is not well suited for our data situation as it does not allow the flexibility to potentially cluster together high-resolution areal units from groups that are not clustered together at the low-resolution level. In Section~\ref{sec:methods}, we will develop a nested Hierarchical Dirichlet Process (nHDP) that combines elements of both the hierarchical Dirichlet process and the nested Dirichlet process in order to simultaneously cluster areal units across multiple levels of resolution.

There has been related work on models for clustering outside of spatial data situations, such as hierarchical topic models \citep{blei2010nested,paisley2014nested,nguyen2014bayesian} and entity topic models \citep{tekumalla2015nested} 
for text data as well as hierarchical clustering of neurological data \citep{wulsin2016nonparametric}, microbiome data \citep{denti2021common} and educational data \citep{beraha2021semi}.   Although \cite{paisley2014nested} also refer to their framework as a nested hierarchical Dirichlet process, the models developed by \cite{blei2010nested} and \cite{paisley2014nested} are designed for a different data context that has an unbounded number of resolution levels.  The models developed by \cite{denti2021common} and \cite{beraha2021semi} combine the nested Dirichlet process with models suitable for grouped data, but with a different approach from the nested HDP. \cite{agrawal2013nested}, \cite{wulsin2016nonparametric} and \cite{nguyen2014bayesian} apply a similar approach in different data contexts, such as the particular case where additional group-level data is available \citep{nguyen2014bayesian}.  \cite{lijoi2022flexible} focuses on the theoretical aspects of multi-level clustering models.  This situation has also been addressed by optimization approaches that extend K-means clustering to distributional spaces using the Wasserstein distance \citep{ho2017multilevel,huynh2019efficient}.

We also note that our problem of multi-level clustering is different from \textit{hierarchical clustering}, where one set of units is recursively divided into clusters, creating several nested partitions of the same data, with clusters within clusters. In our context instead, multiple sets of units exists, one for each resolution level, with a fixed hierarchical structure connecting units in different levels; the clustering finds one partition for each resolution, with not necessarily nested clusters.

The rest of this paper is organized as follows.  We introduce our nested Hierarchical Dirichlet Process (nHDP) approach to simultaneously cluster areal units across multiple levels of resolution in Section~\ref{sec:methods}.  We evaluate the operating characteristics of our nHDP model across a variety of synthetic data situations in Section~\ref{sec:simulation}.   In Section~\ref{sec:application}, we apply our nHDP approach to crime density within areal units at various levels of resolution in the city of Philadelphia.  We examine several interesting findings in detail and then conclude with a brief discussion in Section~\ref{sec:discussion}.

\section{Multi-resolution Clustering Methodology} \label{sec:methods}

\subsection{Areal Units at Multiple Levels of Resolution}

Consider a geographic area of interest $\mathcal{A}$ which is subdivided into areal units at multiple levels of resolution.  For notational simplicity, we consider two granularity levels, a {\it coarser} (or {\it lower}) resolution and a {\it finer} or ({\it higher}) resolution, although our methodology can also be extended to more than two levels of resolution.  

At the lower resolution level, $\mathcal{A}$ is divided up into $L$ areal units that we denote as $\mathcal{L} = \{ A_1, A_2, \ldots, A_{L} \}$.  
The higher resolution level is a further division of $\mathcal{A}$ that is {\it nested} within the lower resolution level.  Each of the lower resolution areal units $A_\ell$ is divided up into $n_\ell$ higher resolution units: $A_\ell = \cup_{h=1}^{n_\ell} A_{\ell,h}$.   We denote the entire collection of the higher resolution areal units as $\mathcal{H} = \{ A_{1,1}, \ldots, A_{1,n_1}, \ldots, A_{L,1}, \ldots, A_{L,n_L} \}$.  

As part of our study of crime dynamics in Philadelphia, we want to use a model that clusters areal units simultaneously at both the lower and higher levels of resolution.  In other words, part of our modeling process will be simultaneously estimating a latent partition $\gamma^{\mathcal{L}}$ of the areal units $\mathcal{L}$ at the lower level of resolution and a latent partition $\gamma^{\mathcal{H}}$ of the areal units $\mathcal{H}$ at the higher level of resolution.  

We want to emphasize that although the higher resolution areal units $\mathcal{H}$ are nested in the lower resolution areal units $\mathcal{L}$, we do not want to restrict ourselves to estimating latent partitions $\gamma^{\mathcal{H}}$ that are nested within $\gamma^{\mathcal{L}}$.   As we describe in our introduction, we want the flexibility to potentially cluster high resolution units together even if their corresponding lower resolution units are not clustered together.  

We will be clustering areal units at both resolution levels based upon the observed crime density $y_{\ell,h}$ in each areal unit $A_{\ell,h}$ at the higher resolution level, though we will also consider aggregated crime density at the lower resolution level.  We denote with $\theta_{\ell,h}$ the parameters underlying the observed crime density $y_{\ell,h}$ in each higher resolution areal unit $A_{\ell,h}$.   We take a Bayesian non-parametric approach where the parameters $\theta_{\ell,h}$ are modeled as a latent discrete mixture where we say that $A_{\ell,h}$ and $A_{\ell',h'}$ are in the same cluster if and only if $\theta_{\ell,h} = \theta_{\ell',h'}$.   We will postpone a more detailed discussion of the model for the observed crime data in Philadelphia until we have reviewed different Bayesian nonparametric prior distributions for latent mixture models. 

\subsection{Prior Distributions for Latent Mixtures and Partitions}

Our modeling goal is to estimate latent partition of the areal units $A_{\ell,h}$ at the higher resolution level based on unit-specific crime density parameters $\theta_{\ell,h}$, while also simultaneously estimating a latent partition of the areal units $A_\ell$ at the lower level of resolution (within which the higher resolution units $A_{\ell,h}$ are nested).   

There are many common Bayesian non-parametric approaches for estimating a partition of unobserved parameters $\theta_{\ell,h}$.  Most of these approaches are built around the Dirichlet Process (DP) prior that was first described by \cite{ferguson1973bayesian}.   The Dirichlet Process is a distribution over random probability distributions that is characterized by a concentration parameter $\alpha_0 >0$ and a base distribution $H_0$. 

Even though the base distribution $H_0$ is typically continuous, a realization $G$ from the Dirichlet Process is almost surely discrete and can be written as $G = \sum p_k \delta_{\theta^*_k}$, where $\delta_{\theta^*_k}$ represent the atoms of $G$ and $p_k$ the probability associated with $\theta^*_k$.  According to the {\it stick-breaking construction} of $G$ \citep{sethuraman1994constructive}, the atom locations $\theta^*_k$ are i.i.d. random variables distributed according to $H_0$ and the probabilities $p_k = b_k \prod_{j=1}^{k-1}(1-b_j)$ where $b_j \overset{iid}{\sim} {\rm Beta}(1, \alpha_0)$.

So if we assumed a DP prior for the high resolution crime density parameters, $\theta_{\ell,h} \sim G$ with $G \sim DP(\alpha_0, H_0)$, then the atoms of $G$ would form a latent partition $\gamma^{\mathcal{H}}$ of the high resolution areal units with clusters corresponding to the $K$ unique atoms of $G$, i.e. $\gamma^{\mathcal{H}} = \{ A_{\ell,h}: \theta_{\ell,h} = \theta^*_{k} \}$.  The distribution of clusters formed by a realization from a Dirichlet process prior is also called the {\it Chinese Restaurant Process} \citep{aldous1985exchangeability}.  However, this Dirichlet process prior formulation ignores the natural groupings of the high resolutions areal units $A_{\ell,h}$ formed by their nested structure within the lower resolution units $A_{\ell}$.   

The Hierarchical Dirichlet Process (HDP) of \cite{teh2006hierarchical} is an extension of the Dirichlet Process for grouped data in which each group $\ell$ of units is assigned a group-specific discrete measure $G_\ell$ that is a realization from a Dirichlet Process with a discrete base measure $G_0$. The discrete base measure $G_0$ is itself a realization from a Dirichlet Process.  So in the HDP formulation, the high resolution crime density parameters, $(\theta_{\ell,h})_{h=1}^{n_\ell} \overset{iid}{\sim} G_\ell$ where $(G_\ell)_{\ell=1}^L \overset{iid}{\sim} DP(\alpha_1, G_0)$ with $G_0 \sim DP(\alpha_0, H_0)$.  The distribution of clusters formed by a realization from a Dirichlet process prior is also called the {\it Chinese Restaurant Franchise} \citep{teh2006hierarchical}.

The discreteness of the base measure $G_0$ implies that all groups $G_\ell$ share the same set of atoms $\theta^*_{k}$ of $G_0$, and so areal units $A_{\ell_1,h}$ and $A_{\ell_2,h}$ from two different low resolution groups $\ell_1$ and $\ell_2$ can be clustered together in a latent partition $\gamma^{\mathcal{H}}$ of the high resolution areal units.  In contrast to the DP, this HDP formulation does acknowledge differences between the low resolution areal units $A_{\ell}$ through the group-specific discrete distributions $G_\ell$.  However, the HDP does not capture any similarities between the low resolution areal units $A_{\ell}$ nor does it produce a partition $\gamma^{\mathcal{L}}$ of the low resolution areal units. See the top panel of Figure~\ref{fig:plate} for representation of this model in the plate notation.

A different extension of the Dirichlet process that would allow simultaneous clustering of both the low resolution areal units $A_{\ell}$ and high resolutions areal units $A_{\ell,h}$ is the nested Dirichlet Process (nDP) of \cite{rodriguez2008nested}.   Similar to the HDP, the nDP is designed for grouped data in which the parameters for each group $\ell$ of units is assigned a group-specific discrete prior measure, $(\theta_{\ell,h})_{h=1}^{n_\ell} \overset{iid}{\sim} G_\ell$.    However, in the nDP formulation, these group-specific measures $G_\ell$ are independently generated from a {\it nested} Dirichlet process: $(G_\ell)_{\ell=1}^L \overset{iid}{\sim} Q$ where $Q \sim DP(\alpha_1, DP(\alpha_0, H_0))$. 

Since $Q$ is a realization from a Dirichlet process, it can be written as $Q = \sum p_j \delta_{Q^*_j}$, where the atoms $Q^*_j$ are themselves each separate realizations from a Dirichlet process $Q^*_j \sim DP(\alpha_0, H_0)$.  The discreteness of $Q$ means that two low resolution areal units $A_{\ell_1}$ and $A_{\ell_2}$ can share the same base distribution atom $Q^*_j$, which generates a partition $\gamma^{\mathcal{L}}$ of the low resolution areal units.  In other words, if $G_{\ell_1} = G_{\ell_2} = Q^*_j$ for some $j$, then the low resolution areal units $A_{\ell_1}$ and $A_{\ell_2}$ are clustered together in $\gamma^{\mathcal{L}}$.   

Since each atom $Q^*_j$ of $Q$ is itself a realization from the Dirichlet process, it can be written as $Q^*_j = \sum_k p_{jk} \delta_{\theta^*_{j,k}}$ where $\theta^*_{j,k} \overset{iid}{\sim} H_0$.   This discreteness of each $Q^*_j$ also means that two high resolution areal units $A_{\ell_1,h_1}$ and $A_{\ell_2, h_2}$ (contained within low resolution units with the same $Q^*_j$) can share the same atoms which generates a partition $\gamma^{\mathcal{H}}$ of the high resolution areal units.  In other words, if $G_{\ell_1} = G_{\ell_2} = Q^*_j$ and if $\theta_{\ell_1,h_1} = \theta_{\ell_2,h_2} = \theta^*_{j,k}$ for some $k$, then the high resolution areal units $A_{\ell_1,h_1}$ and $A_{\ell_2,h_2}$ are clustered together in $\gamma^{\mathcal{H}}$.  See the bottom left panel of Figure~\ref{fig:plate} for a plate diagram of this model.

However, as shown by \cite{camerlenghi2018latent}, this nDP formulation imposes a restriction on the high and low resolution partitions. The two partitions induced by the nested Dirichlet Process are themselves {\it nested}: a cluster in the partition $\gamma^{\mathcal{H}}$ contains high resolution units from different low resolution units only if those low resolution units are clustered together in $\gamma^{\mathcal{L}}$.  The reason for this behavior lies in the fact that the atoms $Q^*_j$ of $Q$ are independent realizations from a Dirichlet Process and thus their supports do not overlap almost surely.  If two low resolution areal units $A_{\ell_1}$ and $A_{\ell_2}$ are not in the same cluster, then $G_{\ell_1} = Q^*_{j}$ and $G_{\ell_1} = Q^*_{j^\prime}$ where the atoms of $Q^*_{j}$ and $Q^*_{j^\prime}$ are almost surely different since their own atoms $\theta^*_{j,k}$ are generated independently from a non-atomic base measure $H_0$ and so there is zero probability that $\theta^*_{j,k}$ is equal to $\theta^*_{j^\prime,k^\prime}$ for any $k$ and $k^\prime$.    

This property is quite restrictive in our application to clustering neighborhoods based on crime (and multi-resolution clustering more generally).  We  want to allow for the possibility that smaller areas (high resolution units) in Philadelphia share similar crime density even if their larger regions (low resolution units) do not share similar crime density.  For example, we anticipate that two adjacent low resolution units could belong to separate clusters because of their overall crime density, but could contain high resolution units on either side of their shared border that display a similar crime density and hence we would want to cluster them together.  Since the nested Dirichlet process only allows for nested partitions, it can not accommodate these types of data situations.

\subsection{The Nested Hierarchical Dirichlet Process for Multi-resolution Clustering}\label{chap4:sec:nHDP}

To allow extra flexibility to accommodate the types of data situations mentioned above, we have developed a new prior formulation for multi-resolution clustering that combines aspects of both the nested Dirichlet Process of \cite{rodriguez2008nested} and the Hierarchical Dirichlet Process of \cite{teh2006hierarchical}. 

We continue to consider the underlying parameters $(\theta_{\ell,h})_{h=1}^{n_\ell}$ of crime density in our high resolution areal units $A_{\ell,h}$ as grouped random variables where the groups $\ell$ represent the low resolution units $A_{\ell}$ within which the high resolution areal units are nested.   The parameters $(\theta_{\ell,h})_{h=1}^{n_\ell}$ within each group have their own group-specific priors $G_\ell$ which are realizations from a common prior measure $Q$ shared by all groups.    Similar to the nested Dirichlet Process, $Q$ is a discrete distribution, $Q = \sum p_j \delta_{Q^*_j}$, whose atoms $Q^*_{j}$ are themselves discrete distributions and $p_j = b_j \prod_{l=1}^{j-1}(1-b_l)$ where $b_l \overset{iid}{\sim} {\rm Beta}(1, \alpha_2)$.  

However, unlike the nested DP where the atoms $Q^*_{j}$ are realizations from a Dirichlet Process, in our nested hierarchical Dirichlet process (nHDP) the atoms $Q^*_{j}$ are 
realizations from a hierarchical Dirichlet process: $Q^*_j \sim DP(\alpha_1, G_0)$ where $G_0 \sim DP(\alpha_0, H_0)$.   So the atoms $Q^*_j = \sum_k p_{jk} \delta_{\theta^*_{k}}$ where the atoms $\theta^*_{k}$ are realizations from a continuous prior measure $H_0$.   For notational compactness, we write this prior model as $Q \sim nHDP(\alpha_0, \alpha_1, \alpha_2, H_0)$. 

When two group-specific priors $G_{\ell_1}$ and $G_{\ell_2}$ share the same atom $Q^*_{j}$, then the low resolution areal units $A_{\ell_1}$ and $A_{\ell_2}$ are clustered together, which gives a partition $\gamma^{\mathcal{L}}$ of the low resolution areal units.  However, our nHDP formulation does allow for more flexible partitions $\gamma^{\mathcal{H}}$ of the high resolution areal units than the nested Dirichlet Process.  In contrast to the nDP, when two low resolution units are not clustered together, i.e. $G_{\ell_1} = Q^*_{j} \neq G_{\ell_2} = Q^*_{j'}$, the high resolution areal units that they contain can still be clustered together since $Q^*_{j}$ and $Q^*_{j'}$ can share the same atoms $\theta^*_k$ and so there is a positive probability that $\theta_{\ell_1,h_1} = \theta_{\ell_2,h_2}$ which means that $A_{\ell_1,h_1}$ and $A_{\ell_2, h_2}$.   This property is especially important for multiple resolution clustering, as it does not force the partition $\gamma^{\mathcal{H}}$ of high resolution areal units to be nested within the partition $\gamma^{\mathcal{H}}$ of low resolution areal units as discussed at the end of the previous subsection.

In the bottom right panel of Figure~\ref{fig:plate} we represent the nested hierarchical Dirichlet process with plate notation. The diagram shows that the low-resolution group-specific measures $G_\ell$ are drawn from a discrete measure $Q$ in the nHDP, similar to the nDP, but that the atoms of $Q$ differ between these two models.  In Section~S\SMfigure~of our supplementary materials, we provide another graphical representation that highlights the difference in atom locations in the $Q^*_j$ of the nDP and  nHDP. 

Compared to some recent approaches, the model proposed in \cite{camerlenghi2018latent} is different from the nHDP as it considers only the case of two groups and extends the nDP by writing the group-specific measures $G_\ell$ as a mixture of a shared distribution and a group specific DP. 
The model proposed in \cite{denti2021common} constructs the $Q_j^*$ with weights distributed according to a stick-breaking prior (like the DP), but sharing the same atoms. This construction avoids the restriction on nested partitions imposed by the nDP, but does not induce any similarity in the weights of the different $Q^*_j$, losing any borrowing of strength across groups that are not in the same cluster. \cite{beraha2021semi} propose a similar model that modifies the $Q_j^*$ to be drawn from a DP where the base distribution is a mixture between a discrete and a continuous distribution. This approach also circumvents the restriction to nested partitions by creating group-specific distributions that share some atoms but are allowed to have different supports. This feature might be better suited for testing equality of group distributions, but makes the model much more complex for our purpose of multi-resolution clustering.

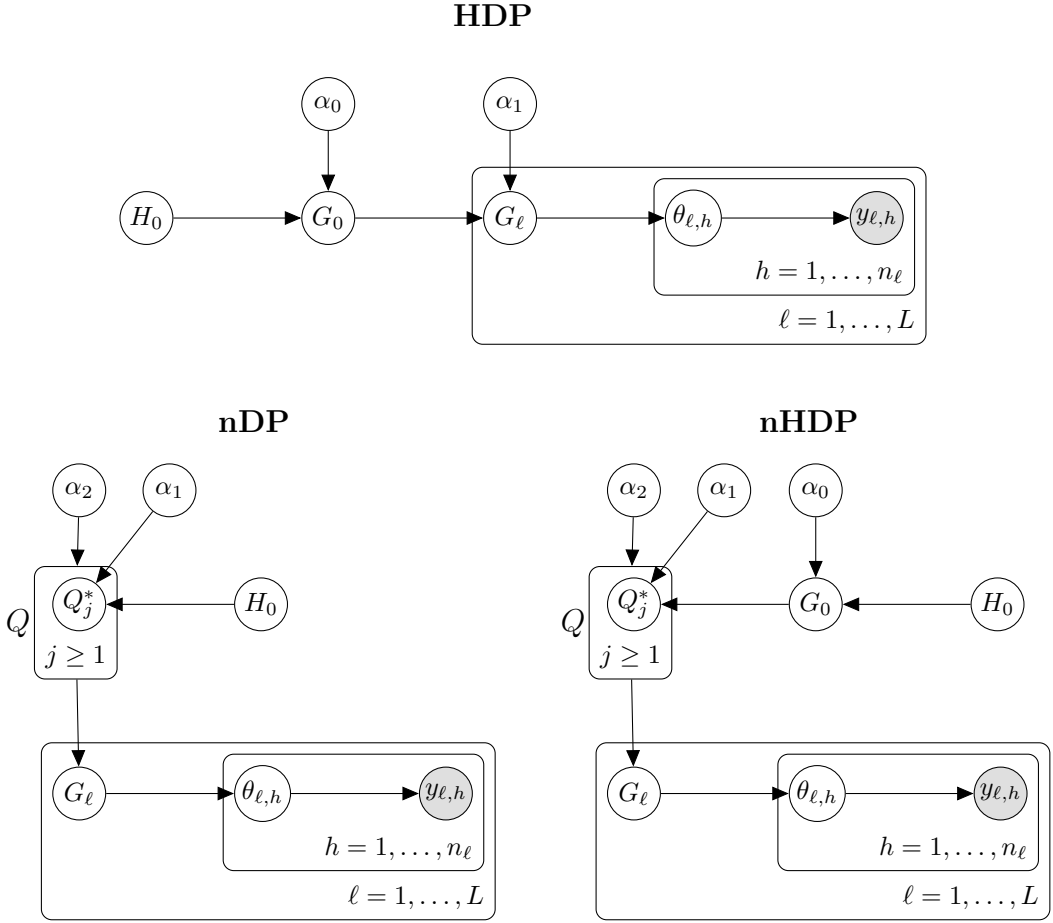
\begin{figure}[H]
\centering
\begin{tikzpicture}[x=1.7cm,y=1.8cm]


  \node[obs]                   (Y)      {$y_{\ell,h}$} ; %
  \node[latent, left=of Y]    (theta)      {$\theta_{\ell,h}$} ; %

  \edge {theta} {Y}
  \plate {plate1} { %
    (Y)
    (theta)
  } {$h = 1, \ldots, n_\ell$}; %

  \node[latent, left=of theta]    (Gl)      {$G_{\ell}$} ; %
  \edge {Gl} {theta}

  \plate {plate2} { %
    (plate1)
    (Gl)
  } {$\ell = 1, \ldots, L$}; %

  \node[latent, left=of Gl] (G0) {$G_0$} ;
  \edge {G0} {Gl};

  \node[latent, left=of G0] (H0) {$H_0$} ;
  \edge {H0} {G0};

  \node[latent, above=of G0, yshift=-1cm] (alpha0) {$\alpha_0$} ;
  \edge {alpha0} {G0};

  \node[latent, above=of Gl, yshift=-1cm] (alpha1) {$\alpha_1$} ;
  \edge {alpha1} {Gl};

  \node at (-3,1.5) {\textbf{HDP}};
\end{tikzpicture} \\
\vspace{20pt}

\begin{tikzpicture}[x=1.7cm,y=1.8cm]


  \node[obs]                   (Y)      {$y_{\ell,h}$} ; %
  \node[latent, left=of Y]    (theta)      {$\theta_{\ell,h}$} ; %

  \edge {theta} {Y}
  \plate {plate1} { %
    (Y)
    (theta)
  } {$h = 1, \ldots, n_\ell$}; %

  \node[latent, left=of theta]    (Gl)      {$G_{\ell}$} ; %
  \edge {Gl} {theta}

  \plate {plate2} { %
    (plate1)
    (Gl)
  } {$\ell = 1, \ldots, L$}; %

  \node[latent, above=of Gl] (Qjstar) {$Q^*_{j}$} ;

  \plate {plate3} { %
    (Qjstar)
  } {$j \geq 1$};
  \edge {plate3} {Gl};

  \node[left=of plate3, xshift=+1.1cm] (Q) {$Q$} ;

  \node[latent, right=of Qjstar] (H0) {$H_0$} ;
  \edge {H0} {Qjstar};

  \node[latent, above=of Qjstar, yshift=-1cm, xshift=+1.2cm] (alpha1) {$\alpha_1$} ;
  \edge {alpha1} {Qjstar};

  \node[latent, above=of Qjstar, yshift=-1cm] (alpha2) {$\alpha_2$} ;
  \edge {alpha2} {plate3};

   \node at (-1.5,2.75) {\textbf{nDP}};
\end{tikzpicture}
~ \hspace{5pt}
\begin{tikzpicture}[x=1.7cm,y=1.8cm]


  \node[obs]                   (Y)      {$y_{\ell,h}$} ; %
  \node[latent, left=of Y]    (theta)      {$\theta_{\ell,h}$} ; %

  \edge {theta} {Y}
  \plate {plate1} { %
    (Y)
    (theta)
  } {$h = 1, \ldots, n_\ell$}; %

  \node[latent, left=of theta]    (Gl)      {$G_{\ell}$} ; %
  \edge {Gl} {theta}

  \plate {plate2} { %
    (plate1)
    (Gl)
  } {$\ell = 1, \ldots, L$}; %

  \node[latent, above=of Gl] (Qjstar) {$Q^*_{j}$} ;

  \plate {plate3} { %
    (Qjstar)
  } {$j \geq 1$};
  \edge {plate3} {Gl};

  \node[left=of plate3, xshift=+1.1cm] (Q) {$Q$} ;

  \node[latent, right=of Qjstar] (G0) {$G_0$} ;
  \edge {G0} {Qjstar};

  \node[latent, right=of G0] (H0) {$H_0$} ;
  \edge {H0} {G0};

  \node[latent, above=of G0, yshift=-1cm] (alpha0) {$\alpha_0$} ;
  \edge {alpha0} {G0};

  \node[latent, above=of Qjstar, yshift=-1cm, xshift=+1.2cm] (alpha1) {$\alpha_1$} ;
  \edge {alpha1} {Qjstar};

  \node[latent, above=of Qjstar, yshift=-1cm] (alpha2) {$\alpha_2$} ;
  \edge {alpha2} {plate3};

  \node at (-1.5,2.75) {\textbf{nHDP}};
\end{tikzpicture}

\caption[Plate diagram of the nHDP.]{Plate diagram of the HDP, nDP and nHDP models. The realizations of low-resolution distributions $G_\ell$ affect the realization of the high-resolution mixture components $\theta_{\ell,h}$, which affect the distribution of the data $y_{\ell,h}$. In the HDP model the $G_\ell$ are directly realization of the Hierarchic Dirichlet Process. In the nDP and in the nHDP models, the $G_\ell$ are drawn from the discrete measure $Q$, whose atoms are realizations respectively from a Dirichlet Process and a Hierarchical Dirichlet Process.
}
\label{fig:plate}
\end{figure}

Our nested hierarchical Dirichlet process induces particular distributions over the partitions $\gamma^{\mathcal{L}}$ and $\gamma^{\mathcal{H}}$ of the low resolution and high resolution areal units respectively, which we briefly describe below with proofs given in Section~S\SMproofs~of the Supplementary Materials.

\begin{prop}
The marginal prior distribution induced by the $nHDP(\alpha_0, \alpha_1, \alpha_2, H_0)$ on the partition $\gamma^{\mathcal{L}}$ of low resolution areal units is the Chinese Restaurant Process:
$$p(\gamma^{\mathcal{L}}) = CRP(\alpha_2). $$
\end{prop}

\begin{prop}
The prior distribution induced by the $nHDP(\alpha_0, \alpha_1, \alpha_2, H_0)$ on the partition $\gamma^{\mathcal{H}}$ of the high resolution areal units is a Chinese Restaurant Franchise where the groups are defined by the clusters in the partition $\gamma^{\mathcal{L}}$ of the low resolution areal units:
$$ p(\gamma^{\mathcal{H}} \, \vert \, \gamma^{\mathcal{L}}) = CRF(\gamma^{\mathcal{H}} \, \vert \, \alpha_0, \alpha_1, \gamma^{\mathcal{L}}).$$
\end{prop}

Now that we have specified our nested Hierarchical Dirichlet process prior for the underlying crime density parameters, we complete our model by specifying a distribution for our observed crime data.  In particular, we will be modeling the crime density (number of crimes divided by area) in each areal unit.  Within each high resolution areal unit $A_{\ell,h}$, we assume that the observed crime density $y_{\ell,h}$ is normally distributed with mean equal to $\theta_{\ell,h}$ and common variance $\sigma^2$. This is equivalent to model the overall distribution of crime density as a mixture of normals.
Note that $\sigma^2$ represents the variance within each high resolution cluster since since all high resolution areal units within a cluster have the same mean crime density $\theta_{\ell,h}$.  

As described above, we model the underlying crime density parameters $\theta_{\ell,h}$ as a latent discrete mixture with our proposed nested hierarchical Dirichlet process prior that induces partitions $\gamma^{\mathcal{H}}$ and $\gamma^{\mathcal{L}}$ on the high resolution and low resolution areal units respectively.  After rescaling the data, we set the base measure $nHDP(\alpha_0, \alpha_1, \alpha_2, H_0)$ as a normal distribution centered at zero with variance ${k_0}^{-1}\sigma^2$.  We use an inverse-gamma prior distribution for the common variance parameter $\sigma^2$.

Below we provide the full set of distributions in our model for observed crime density $y_{\ell,h}$ in the high resolution areal units $A_{\ell,h}$: 
\begin{align}
\begin{split}
y_{\ell,h} \, \vert \, \theta_{\ell,h} &\sim {\rm Normal} \, (\theta_{\ell,h} \, , \, \sigma^2)\\
\theta_{\ell,h} \, \vert \, G_\ell &\sim G_\ell\\
G_\ell \, \vert \, Q &\sim Q \\
Q &\sim nHDP \, (\alpha_0, \alpha_1, \alpha_2, H_0)\\
H_0 &= {\rm Normal} \, (0 \, , \, \frac{\sigma^2}{k_0})\\
\sigma^2 &\sim {\rm Inv-Gamma} \, (\beta_0 \, , \, \beta_1).
\end{split}
\label{eq:mod1}
\end{align}

We additionally consider truncated normal priors for $\alpha_0, \alpha_1, \alpha_2$.
The hyperparameters $(\beta_0,\beta_1)$ that determine the prior distribution of $\sigma^2$ are chosen to achieve the desired level of data variation within a cluster; the value $k_0$ is specified so that the between-cluster distribution covers the range of the data.
We discuss hyperparameter choices in more detail in Section~\ref{sec:application} as well as Sections~S\SMsimulations~and~S\SMphilly~of our supplementary materials.

\subsection{Posterior inference} \label{MCMCsampler}

We now describe a Markov Chain Monte Carlo sampling scheme for estimating the posterior distribution of the Nested Hierarchical Dirichlet Process mixture model. For simplicity we consider the case where $F$ and $H$ are conjugate distributions, so that the mixture component parameters $\theta_k^*$ can be integrated out. The non-conjugate case can be similarly derived using the ideas in \cite{jain2007splitting}.

Similarly to \cite{zuanetti2018clustering}, we consider a sampling scheme for the marginalized nHDP in which only the latent partitions $\gamma^{\mathcal{L}}$ and $\gamma^{\mathcal{H}}$ are iteratively sampled. 
In a first step we sample the high-resolution partition $\gamma^{\mathcal{H}}$ given the low-resolution partition $\gamma^{\mathcal{L}}$ and the data; this can be carried out with one of the posterior sampling schemes for the HDP, for example those described in \cite{teh2006hierarchical}. 
The second step to sample $\gamma^{\mathcal{L}}$ given $\gamma^{\mathcal{H}}$ requires a more complex procedure: since the 
hierarchical structure
of $\gamma^{\mathcal{H}}$ 
and its dimensionality 
changes when we change $\gamma^{\mathcal{L}}$, this problem falls into the category of models described by \cite{green1995reversible}. However, the reversible-jump framework is difficult to apply in this setting due to the discreteness and complexity of the space of $\gamma^{\mathcal{H}}$. Instead, we use the Metropolis-Hasting algorithm with a proposal that allows for movements in the product space of $(\gamma^{\mathcal{L}}, \gamma^{\mathcal{H}})$.  Specifically, we propose a move for $\gamma^{\mathcal{L}}$ and,  conditional on this proposed value, we propose a value for $\gamma^{\mathcal{H}}$ that is consistent with the new configuration of $\gamma^{\mathcal{L}}$.

\paragraph{Sampling $\gamma^{\mathcal{H}}$.} As highlighted in Proposition~2, the groups of low-resolution areal units $A_\ell$ constituting the hierarchy of high-resolution units $A_{\ell,h}$ is known given the low-resolution partition $\gamma^{\mathcal{L}}$.  Using the Chinese Restaurant Franchise (CRF) metaphor, the division of costumers (high resolution areal units) into restaurants (clusters of the low-resolution units) is known and fixed.  
Thus, given $\gamma^{\mathcal{L}}$, we can sample the high-resolution partition using standard CRF sampling schemes. We represent $\gamma^{\mathcal{H}}$ with the assignment of costumers to \textit{tables} within each restaurant and the assignment of tables to \textit{dishes} across restaurants \citep{teh2006hierarchical}.   Rather than exploring the partition space using local Gibbs-type updates, we modify the partitions with broader `split-merge' moves that improve mixing \citep{jain2004split}. We extend the algorithm of \cite{wang2012split} to perform split-merge moves not only in the update of the partition of costumers into tables but also in sampling the dishes for those tables.  See Section~S\SMsplitmerge~in our supplementary materials for additional details.

\paragraph{Sampling $\gamma^{\mathcal{L}}$.} 

Since the clusters of $\gamma^{\mathcal{L}}$ define the division of customers into restaurants, particular caution should be used when sampling from the conditional posterior of $\gamma^{\mathcal{L}}$ given $\gamma^{\mathcal{H}}$. In fact, when the restaurants structure is changed because the clusters of $\gamma^{\mathcal{L}}$ are, tables need to be rearranged and so do the dishes associated with them.
Moreover for each value of $\gamma^{\mathcal{H}}$, the partition of groups $\gamma^{\mathcal{L}}$ is uniquely identified, meaning that the conditional posterior $p(\gamma^{\mathcal{L}} \vert \gamma^{\mathcal{H}}, \x)$ is a point mass.  In acknowledgement of these issues, we use a Metropolis-Hasting step that proposes a new configuration for the pair $(\gamma^{\mathcal{L}},\gamma^{\mathcal{H}})$.
The proposal distribution is broken down into an initial step proposing a new configuration for $\gamma^{\mathcal{L}}$ and a second step proposing a value for $\gamma^{\mathcal{H}}$ given the sampled configuration of $\gamma^{\mathcal{L}}$.
The first step consists in a split-merge move similar to \cite{jain2004split}'s proposal, while the second step comprises simple rearrangements of the tables in the new restaurant configuration, using the CRF representation with tables and dish assignments.
Further details are provided in Section~S\SMsmnHDP~of our supplementary materials.

The multi-level split-merge algorithm described here has been integrated with the parallel tempering \citep{geyer1991markov} to alleviate mixing problems that are often common in multimodal distributions such as this.
The code has been implemented in C++ and R, and is available online at 
\if1\blind
{
\url{https://github.com/cecilia-balocchi/multiresolution_clustering}.
} \fi
\if0\blind
{
 \url{https://anonymous.4open.science/r/multiresolution_clustering-56D4} 
 (link and repository have been anonymized for double-blindness).
  } \fi


\section{Synthetic Data Evaluation} \label{sec:simulation}

In this section, we evaluate the performance of our nested hierarchical Dirichlet Process model by comparing to other methods in synthetic data generated under various conditions.  

We first consider a situation where the synthetic data is generated from one of six different normal mixtures $F_j$.  Each of these normal mixtures $F_j (y) = \sum_{k=1}^6 w_{jk} \phi ((y - \mu_k)/\sigma)$ are the convolution of Gaussian kernels around the same six equally spaced means $(\mu_1,\ldots,\mu_6)$ but with different weights $\w_j$ on each Gaussian component.   These six different normal mixtures are displayed in Figure~\ref{fig:framework2}.

\begin{figure}[H]
\centering
\includegraphics[width = \textwidth]{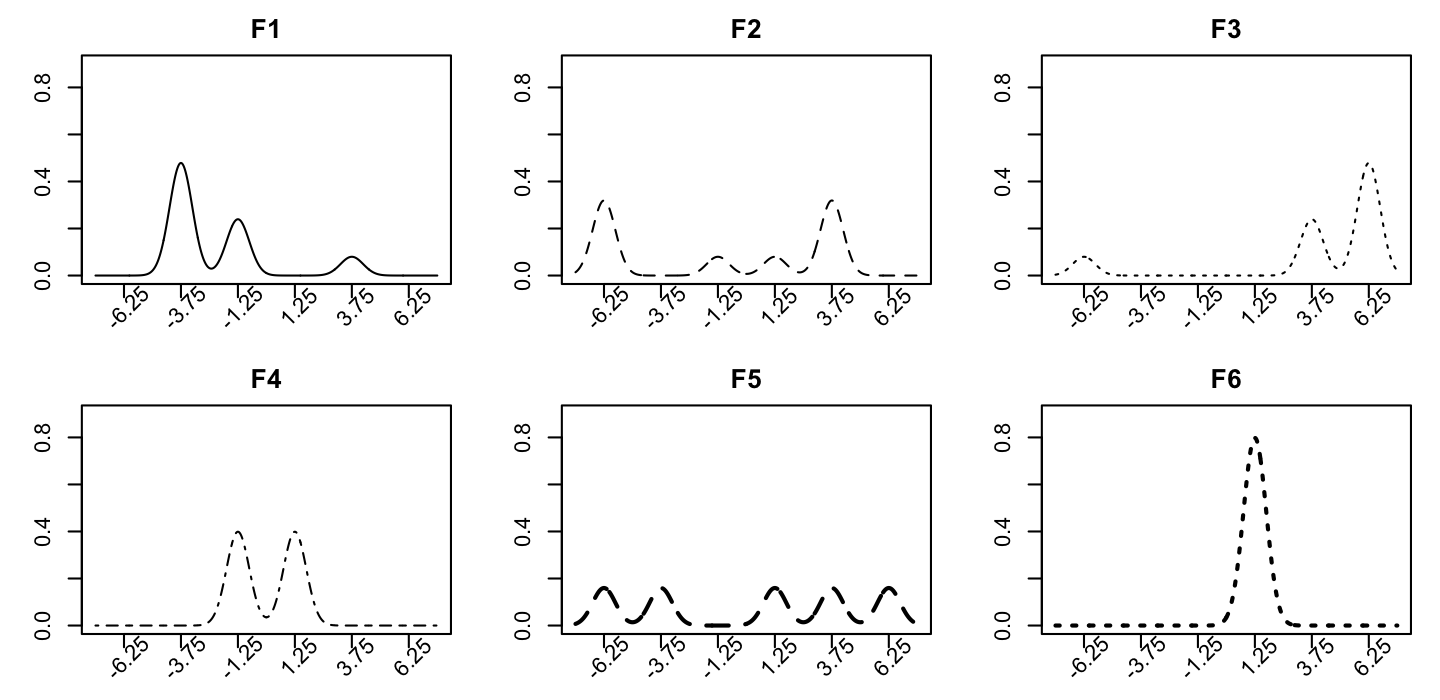} 
\caption{Normal mixtures underlying our synthetic data generation. \label{fig:framework2}}
\end{figure}

We generate datasets that have $L = 25$ low resolution units, where each of these units are uniformly assigned to one of those six normal mixtures.  The true low resolution partition $\gamma^{\mathcal{L}}$ corresponds to the index $j$ of the normal mixture $F_j$ that was assigned to each low resolution unit.  Each of these low resolution units contains the same number $n_\ell$ of high resolution units.  We consider two settings for the number of high resolution units, $n_\ell = 10$ and $n_\ell = 50$, which give a total of $250$ and $1250$ high resolution units across the $L = 25$ low resolution units, respectively.

Within each low resolution unit, each of the $n_\ell$ high resolution units is sampled from the normal mixture $F_j$ that was assigned to that low resolution unit.   The cluster assignment for each high resolution unit is one of the six normal components in $F_j$, sampled with probabilities proportional to the weights $\w_j$ in $F_j$.   So the true high resolution partition $\gamma^{\mathcal{H}}$ is the set of indices corresponding to the normal components sampled for each high resolution unit.  Section~S\SMsimmixt~of the Supplementary Materials provides further details about the synthetic data generation.

We compare the behavior of several methods that are able to simultaneously estimate the two partitions $\gamma^{\mathcal{L}}$ and $\gamma^{\mathcal{H}}$: our proposed nested hierarchical Dirichlet process model (nHDP), the nested Dirichlet process (nDP), and an adaptation of K-means for multi-level clustering.  Specifically, we apply K-means algorithm first to the high resolution data, while using the silhouette method \citep{rousseeuw1987silhouettes} to choose the number of clusters.  The partition created by this clustering of high resolution units is used to create a vector of high resolution cluster proportions within each low resolution unit.  We then run the K-means algorithm (again with the silhouette method) on these vectors of cluster proportions to create a partition of the low resolution units.     

We use the nDP implementation of \cite{zuanetti2018clustering}. For both the nHDP and nDP, we ran two MCMC chains for twelve thousand iterations, with convergence determined to occur after the two thousand iterations which were then discarded as burn-in. 

We evaluate the performance of each method in terms of both parameter estimation as well as recovery of the true partitions $\gamma^{\mathcal{L}}$ and $\gamma^{\mathcal{H}}$ underlying both resolutions of units.   

Our measure of parameter estimation at the high resolution level is the root mean squared error (RSME) of the estimated means, $\mu_k$, for each normal component with this estimation being performed seperately within each low resolution unit.    Our measure of parameter estimation at the low resolution level is the RMSE of the estimated overall mean of the normal mixture, $\E(F_j)$, that was assigned to each low resolution unit.  For the nHDP and nDP models, we use posterior mean estimates of these quantities, whereas for our adaptation of K-means, we use mean estimates computed conditionally on the estimated partition.

For our measures of partition recovery, we use the Variation of Information (VI) distance \citep{meilua2007comparing} between the true partition and the estimated partition at each resolution level.  For the MCMC-based nHDP and nDP models, the estimated partitions are found as the minimizers of the VI distance between the sampled partitions as recommended by \cite{wade2018bayesian}.

We compare the performance of the nHDP and nDP models as well as our multi-level adaptation of K-means (km) on our measures of parameter estimation and partition recovery in Figure~\ref{fig:sim}.   Separate rows of performance are shown for data generated with either $n_\ell = 10$ and $n_\ell = 50$ high resolution units within each of our $L = 25$ low resolution units.   In Section~S\SMsimmixt~of our supplementary materials, we see similar results for datasets generated with either smaller ($L = 10$) or larger ($L = 50$) numbers of low resolution units.

\begin{figure}[H]
\centering
\includegraphics[width = \textwidth]{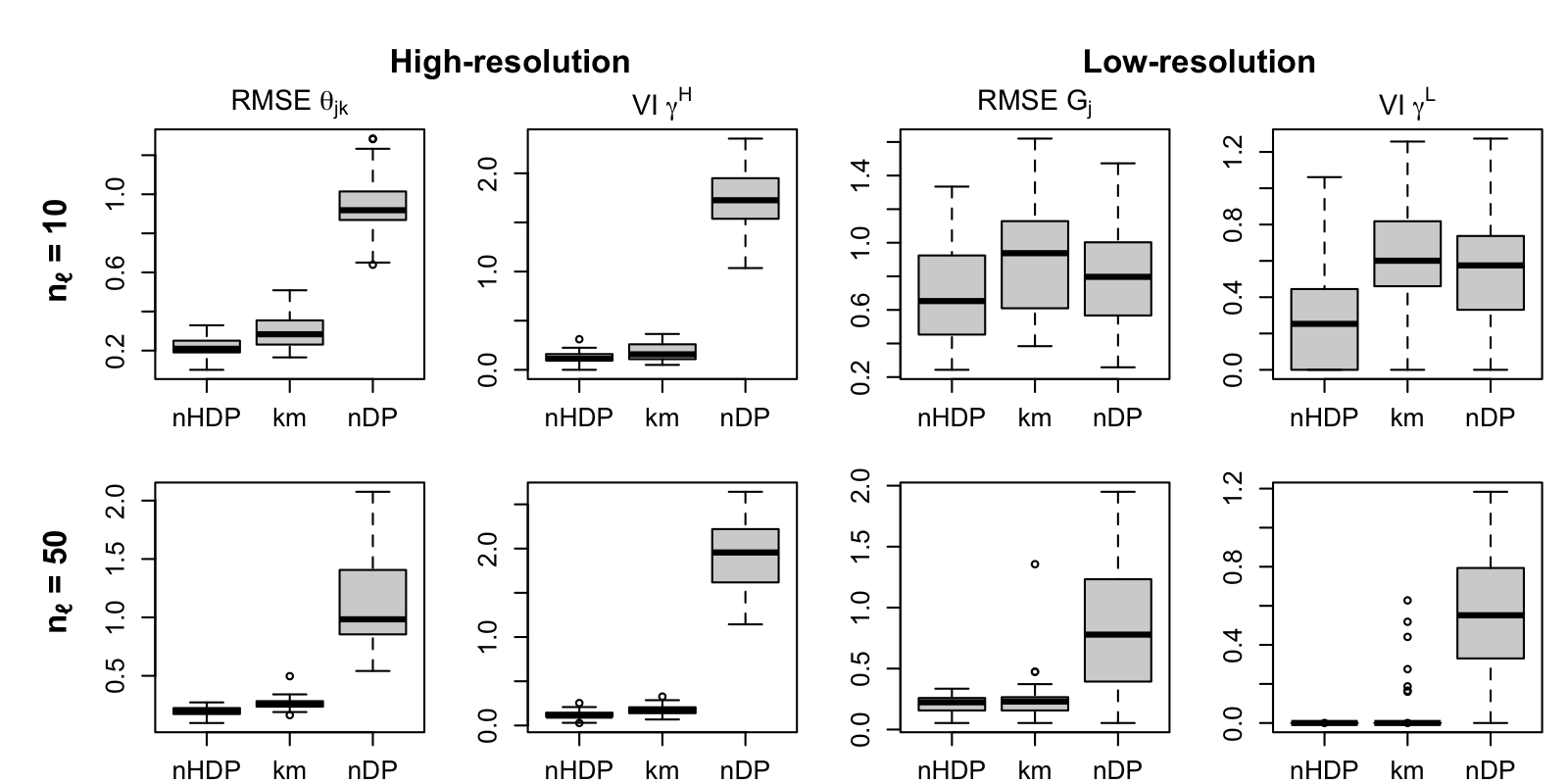} 
\caption{Performance of the nHDP, nDP, and K-means (km) models on our measures of parameter estimation and partition recovery at the high resolution level (left panels) and low resolution level (right panels).  Each dataset was created with $L = 25$ low resolution units that each contain either $n_\ell = 10$ high resolution units (top panels) or $n_\ell = 50$ high resolution units (bottom panels). \label{fig:sim}}
\end{figure}

From Figure~\ref{fig:sim}, we see that our proposed nested hierarchical Dirichlet Process (nHDP) achieves the best performance among the three methods in terms of both parameter estimation and partition recovery.   We also see that the relative performance of the different methods is similar between the tasks of parameter estimate versus partition recovery.  This not surprising in that recovering a partition closer to the truth should allow for better parameter estimation.  However, an estimated partition with too many clusters could overfit the data and still have decent parameter estimation while being a poor estimate of the true partition itself. 

Our adaptation of the K-means algorithm (km) is competitive with the nHDP in several of these data settings but exhibits poor parameter estimation and partition recovery at the low resolution level when there are a smaller ($n_\ell = 10$) number of high resolution units.    An important advantage of our nHDP approach is superior performance in difficult but common data settings in which there are a limited number of high resolution units within each low resolution unit.  

In contrast, the nested Dirichlet process (nDP) is not competitive and shows substantially worse performance across both levels of resolution and regardless of the number of high resolution units per low resolution units.   The nDP method suffers from being restricted to only estimating nested partitions in this data situation where the true partitions are not nested.

In Section~S\SMsimmodel~of the supplementary materials, we see similar results in an additional simulation study where synthetic data was generated from our model \eqref{eq:mod1} which produces less uniform true partitions compared to our data generation above.  

\section{Clustering crime in West Philadelphia}\label{sec:application}

Crime data in Philadelphia are made publicly available by the Philadelphia Police department.  Information is provided on the date, time and GPS location for every reported crime from 2006 to 2018, as well as the crime type.  In this analysis, we focus on only violent crimes which consist of homicides, rapes, robberies and aggravated assaults, according to the definition by the Uniform Crime Reporting program of the FBI.  We calculate the number of violent crimes per year within each US census block group (our high resolution units) and each US census tract (our low resolution units).  We then average the count of violent crimes per year within each areal unit over the 2006-18 time period.  

However, average violent crime counts are difficult to compare directly between the two levels since low resolution counts are on different scale from high resolution counts.  Thus, we convert these counts into a {\it rate} of violent crimes that are more comparable between US census block groups and US census tracts.  While crime rates by residential population are commonly used by criminologists, it has been argued that crimes are often committed by (and against) individuals that do not reside in that particular neighborhood \citep{zhang2007spatial}.   So instead we will focus on violent crime {\it density} per unit area, as suggested by \citet{zhang2007spatial}.

We will focus our analysis on the West Philadelphia region for which we display both raw crime densities and clustering results in Figure~\ref{fig:wp_real_nHDP}.    This part of the city is home to several universities and has large heterogeneity in observed crime densities between different areal units at both levels of resolution.  

\begin{figure}[H]
\centering
\includegraphics[width = 0.95\textwidth]{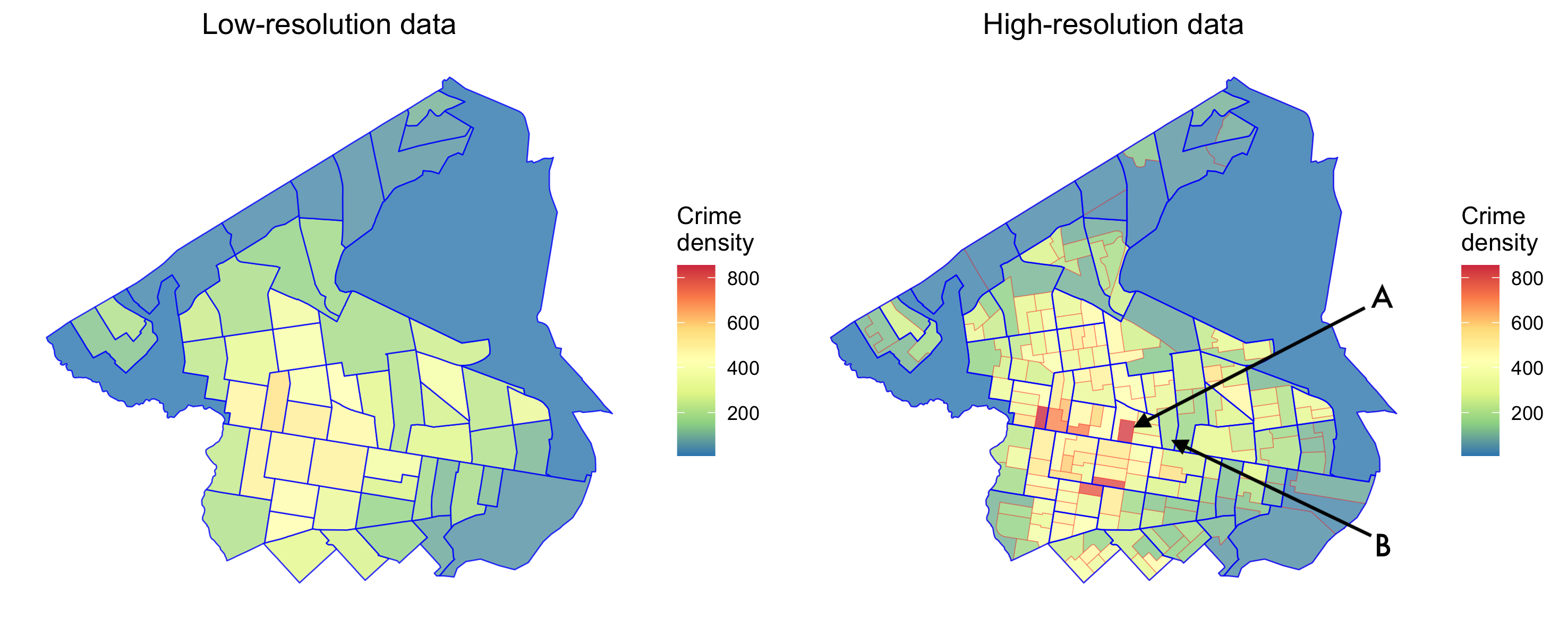} \\
\includegraphics[width = 0.95\textwidth]{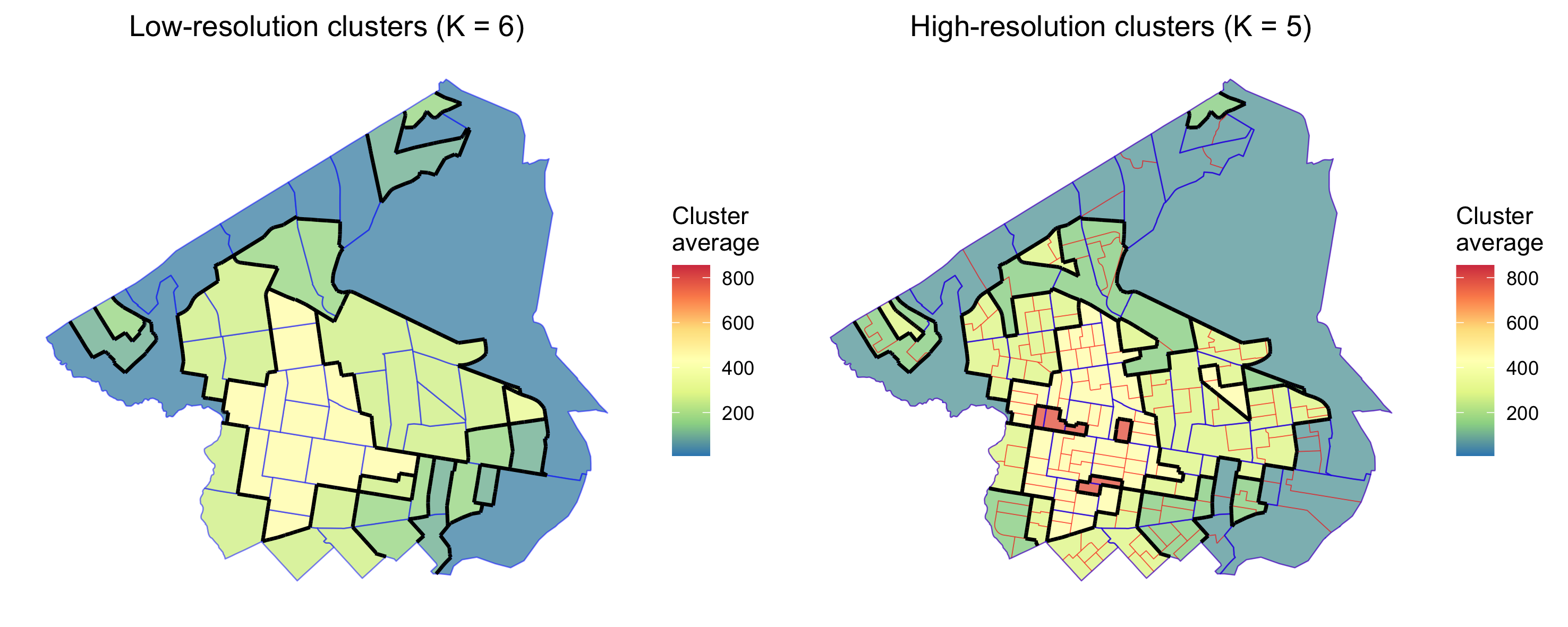} 
\caption{
{\it Top panels}: West Philadelphia divided into 56 census tracts delineated in blue (left) and 202 census block groups delineated in red (right). Each area is colored by the average violent crime density per year (2006-18). 
{\it Bottom panels}: Estimated low resolution (left) and high-resolution (right) partitions from our nHDP model.  Each cluster is colored by the average of the average violent crime density per year within that cluster. }
\label{fig:wp_real_nHDP}
\end{figure}

For example, we compare the US census block group at 52nd St and Market St (labeled A in Figure~\ref{fig:wp_real_nHDP}) which has one of the largest crime densities at 822 violent crimes per squared kilometer versus the US census block group at 48th St and Market St (labeled B in Figure~\ref{fig:wp_real_nHDP}) that is only half a mile away but has one of the lower crime densities at 208 crimes per squared kilometer.  The built environment plays a substantial role in this heterogeneity.  The high crime block group at 52nd St and Market St is the location of a major metro station whereas the low crime block group at 48th St and Market St is the location of a high school and hospital building.   

Examining the raw crime densities in the top panels of Figure~\ref{fig:wp_real_nHDP}), it is important to note that the high crime areas seen in red at the center of the high resolution map would not be detected if only the low resolution data was analyzed.    These maps also show that the patterns in crime levels amongst the high resolution units are not contained within units at the low resolution level, which makes the restriction to nested partitions (that is characteristic of the nested Dirichlet process) ill-suited for this data situation.  

So we focus our analysis of this region on the low and high resolution partitions that we estimate from our nested hierarchical Dirichlet process (nHDP) model, after rescaling the violent crime densities to be centered at zero and have unit variance.   We used a truncated normal prior with mean 2 and standard deviation 1 for $\alpha_0, \alpha_1$ and $\alpha_2$.  We set the values of the hyperparameters $\beta_0$ and $\beta_1$ so that $\sigma^2$ had a prior mean of 0.25 and a prior standard deviation of 0.1.   These values express our prior belief that each cluster covers a range of up to 1 data standard deviations with a within-cluster variation of approximately 0.5 data standard deviations.  Finally, we set $k_0 = 1/10$ to ensure that the base distribution $H_0$ covers the full range of the data.  

We ran two chains the MCMC sampler described in Section~\ref{MCMCsampler} for 50000 iterations, and then removed the first 10000 iterations of each chain as burn-in and thinning the remaining samples to only retain 1 out of every 50 iterations.  The remaining 800 samples per chain were combined and then we  extracted the best high and low resolution partitions from these samples using the method of \cite{wade2018bayesian}.   In Section~S\SMphilly~of the supplementary materials, we provide additional details on our hyperparameter choices and partition estimation from the MCMC samples.

Maps of the best high and low resolution partitions estimated from our nested hierarchical Dirichlet process model are shown in the bottom panels of Figure~\ref{fig:wp_real_nHDP}.   Each cluster is highlighted with thick black borders with the internal color corresponding to the mean violent crime density within that cluster.   The partition of high resolution units (US census block groups) consists of 5 clusters whereas the partition of low resolution units (US census tracts) consists of 6 clusters.  

Overall, the estimated partitions in the bottom panels of Figure~\ref{fig:wp_real_nHDP} give a good visual approximation of the raw crime densities (shown in the top panels of Figure~\ref{fig:wp_real_nHDP}) while offering a simpler and easier to interpret map of violent crime levels in West Philadelphia.  We also note that the clusters of US census block groups are not nested within the clusters of US census tracts which would be a restriction imposed by the nested Dirichlet process.  

We see a strong correspondence between the low and high resolution partitions in the periphery regions of West Philadelphia that have low or medium violent crime densities (blue and darker green).  However, only the map of the high resolution partition enables us to distinguish between the locations with medium (yellow) vs. high (red) violent crime densities in the central region of West Philadelphia.    The red locations with the highest violent crime densities are small and scattered enough to not be detected at all in the low resolution partition or low resolution raw crime density maps in Figure~\ref{fig:wp_real_nHDP}.  

In summary, we see that the partition of low resolution US census tracts are an adequate summary of the low to medium violent crime densities found in the peripheral regions of West Philadelphia.  However, we see an increase in violent crime density towards the center of West Philadelphia, where the cluster of US census block groups with highest violent crime densities are only discovered in the high resolution partition estimated by our model.   

\section{Discussion} \label{sec:discussion}

Estimation of the spatial variation in crime in large cities is a challenging endeavor as patterns in crime density are not necessarily smooth due to physical and social boundaries within urban environments, with the additional complication that different resolutions of data aggregation are available these modeling efforts.   We have addressed these issues by developing a nested hierarchical Dirichlet process (nHDP) model that clusters areal units across multiple levels of resolution simultaneously.   Our approach is more flexible than the popular nested Dirichlet Process model of \citep{rodriguez2008nested} in the sense that the nHDP is not restricted to estimating only nested partitions between different levels of resolution. 

We apply our nHDP approach to the estimation of violent crime density in the city of Philadelphia, with a focus in Section~\ref{sec:application} on the West Philadelphia region which has substantial heterogeneity in violent crime incidence.   We simultaneously estimate partitions of this region at both the lower resolution of US census tracts and the higher resolution of US census block groups.    We find a high similarity between the partitions of US census tracts and US census block groups in the peripheral neighborhoods of West Philadelphia that have the lowest levels of violent crime density.   In the more central area of this region, only the high resolution partition estimated by our model is able to detect a cluster of US census block groups with the highest violent crime densities in West Philadelphia.   These locations are too small and not proximal enough to each other to be found in the partition of the lower resolution US census tracts.  

Our Markov Chain Monte Carlo model implementation involves split-merge moves for updating the partition of lower resolution units.  While split-merge algorithms are known for having good mixing properties, it would be convenient in this high dimensional setting to develop a more direct Gibbs sampling step for updating the low resolution partition.   While the conditional probabilities needed for a Gibbs step cannot be analytically computed in a simple way, it may be possible to use numerical techniques that could lead to more efficient sampling from the posterior distribution over partitions.  That said, even an efficient MCMC algorithm can be limited in high dimensional problems, especially as the number of units or number of resolution levels increase.   Alternative models that offer better scalability should be an object of future research.  The Bayesian Additive Regression Trees \citep{chipman2010bart} model is a promising option in this direction.  

While our nHDP model was developed to address the specific challenges of crime estimation in urban environments, there are many other domains that would benefit from our approach to multi-resolution modeling.    The issue of multiple levels of resolution is ubiquitous whenever data is being aggregated within areal units which is a very common situation in many fields such as epidemiology, ecology and neuroscience.   For example, in the domain of neuroimaging, a multi-resolution clustering model could be used to find similarities between larger brain regions while also detecting finer patterns of behavior among individual voxels or sets of voxels.

\section{Acknowledgements} 

The authors are grateful to James E. Johndrow for helpful comments and suggestions. The first author was supported by the European Research Council (ERC) under the European Union’s Horizon 2020 research and innovation programme under grant agreement No. 817257.  The third author gratefully acknowledges funding from NSF grant DMS-1916245.

\bibliographystyle{apalike}
\bibliography{references}


\newpage
\begin{center}
{\Large {\bf Supplementary Materials for }}

\bigskip

{\Large {\bf ``Clustering Areal Units at Multiple Levels of Resolution to Model Crime in Philadelphia"}}

\bigskip

\end{center}

\setcounter{equation}{0}
\setcounter{figure}{0}
\setcounter{table}{0}
\setcounter{page}{1}
\setcounter{section}{0}
\setcounter{prop}{0}
\makeatletter
\renewcommand{\theequation}{S\arabic{equation}}
\renewcommand{\thefigure}{S\arabic{figure}}
\renewcommand{\bibnumfmt}[1]{[S#1]}
\renewcommand{\citenumfont}[1]{S#1}
\renewcommand\thesection{S\arabic{section}}

\section{Properties of the nested Hierarchical Dirichlet Process}
\label{app:proofs}

\begin{prop}
The marginal prior distribution induced by the $nHDP(\alpha_0, \alpha_1, \alpha_2, H_0)$ on the partition of groups $\gamma^{\mathcal{L}}$ is the Chinese Restaurant Process:
$$p(\gamma^{\mathcal{L}}) = CRP(\alpha_2). $$
\end{prop}
\begin{proof} 
Conditional on $G_0$, $Q$ is a realization of a Dirichlet Process with concentration parameter $\alpha_2$ and with base measure equal to a $DP(\alpha_1, G_0)$; this means that $Q$'s atoms $Q^*_j$, despite sharing the same atoms $\theta^*_k$, are all different {\it a.s.}, because their sequences of weights $(p_{jk})$ are all different {\it a.s.}. Moreover $Q$'s weights $(w_k)$ are generated according to the stick-breaking construction with parameter $\alpha_2$. Thus, conditional on $G_0$, the conditional distribution of $G_\ell$ given $G_1, \ldots, G_{\ell-1}$ follows the P\'olya Urn scheme: 
$G_\ell \vert G_1, \ldots, G_{\ell-1}, G_0 \sim \frac{1}{\ell-1+\alpha_2}\sum_{i=1}^{\ell-1} \delta_{G_i} + \frac{\alpha_2}{\ell-1+\alpha_2} DP(\alpha_1, G_0).$
From this follows that $p(\gamma^{\mathcal{L}} \vert G_0) = CRP(\alpha_2)$. To find the marginal distribution of $\gamma^{\mathcal{L}}$ we need to integrate out $G_0$; for this purpose, note that $G_0$ only affects the distribution of a new observation, i.e. $DP(\alpha_1, G_0)$. Then we just need to show that, marginally on $G_0$, a new observation $Q^*$ is different from the previously observed $G_1, \ldots, G_{\ell-1}$. Note that for any $G_i$, $p(Q^* = G_i \vert G_0) = 0$, because the Dirichlet Process is a non-atomic distribution on the space of probability measures. Since this is true for any $G_0$, $p(Q^* = G_i \vert G_1, \ldots, G_{i-1}) = \int p(Q^* = G_i \vert G_0) p(dG_0 \vert G_1, \ldots, G_{i-1}) = 0$.
\end{proof}

\begin{prop}
The prior distribution induced by the $nHDP(\alpha_0, \alpha_1, \alpha_2, H_0)$ on the partition of observations $\gamma^{\mathcal{H}}$ conditional on the partition of groups $\gamma^{\mathcal{L}}$ is a Chinese Restaurant Franchise distribution, where the groups are defined by the clusters of $\gamma^{\mathcal{L}}$:
$$ p(\gamma^{\mathcal{H}} \vert \gamma^{\mathcal{L}}) = CRF(\gamma^{\mathcal{H}} \vert \alpha_0, \alpha_1, \gamma^{\mathcal{L}})
$$
\end{prop}
\begin{proof}
Let $\gamma^{\mathcal{L}}$ be the partition of low resolution units, where each cluster is formed by the low resolution units associated to the same atom of $Q$, for example let $S^{\mathcal{L}}_k = \{ A_j : G_j = Q^*_k\}$. 
Then the observations corresponding to the low resolution units in the same cluster $S^{\mathcal{L}}_k$ share the same distribution. If we consider the vector of observations from the groups belonging to cluster $S^{\mathcal{L}}_k$, $\Btheta_{S^{\mathcal{L}}_k} = \{ \theta_{\ell,1}, \ldots, \theta_{\ell,n_\ell}: A_\ell \in S^{\mathcal{L}}_k \}$, 
then $\Btheta_{S^{\mathcal{L}}_k} \vert Q^*_k \overset{iid}{\sim} Q^*_k$, with $G^*_k \vert G_0 \sim DP(\alpha_1.G_0)$ for all $k$ and $G_0 \sim DP(\alpha_0, H_0)$.
Thus, conditional on $\gamma^{\mathcal{L}}$, we can divide the $\theta_{\ell,h}$ into the collections defined by the clusters $S^{\mathcal{L}}_k$ and they are distributed according to a Hierarchical Dirichlet Process. 
We can then define $\gamma^{\mathcal{H}}$ by considering the $\theta_{\ell,h}$ that take on the same values across these collections; the distribution of $\gamma^{\mathcal{H}}$ is then described by the Chinese Restaurant Franchise with groups defined by the clusters of $\gamma^{\mathcal{L}}$.
\end{proof}

\section{Algorithm for posterior sampling}
\label{app:smHDP}

\subsection{Split-merge for HDP}
In this section we are going to present a posterior sampling algorithm for the Hierarchical Dirichlet Process. We use the Chinese restaurant franchise representation, described by \cite{teh2006hierarchical}. Instead of using the Gibbs sampling algorithm described by \cite{teh2006hierarchical}, we propose a Split-Merge algorithm for the HDP, extending the work of \cite{jain2004split}.
In the Chinese Restaurant Franchise representation of the HDP, the partition is described by a partition of costumers into tables within each restaurant and a partition of tables into dishes across restaurants.
Let $t_{ji}$ be the table assigned to costumer $i$ in restaurant $j$, with $\bt_j = (t_{ji}: \forall i)$ being the partition of costumers into tables in restaurant $j$, $\bt = (t_{ji}: \forall j,i)$ and $\bt_{-j} = (t_{j'i}: \forall i, j'\neq j)$; moreover let $k_{jt}$ be the dish assigned to table $t$ in restaurant $j$, and $\bk = (k_{jt}:\forall j,t)$ be the partition of tables across restaurants into dishes.

Remember that we can write the likelihood as 
$$p(\y \vert \bt, \bk) = \prod_{k} \int \prod_{j,i: k_{jt_{ji}} = k} p(y_{ji} \vert \phi_k) d\phi_k$$
and that the prior $p(\bk, \bt) = p(\bk \vert \bt) \prod_j p(\bt_j)$, where $p(\bk \vert \bt)$ and $p(\bt_j)$ are Ewens-Pitman prior distributions for partitions. 

Moreover remember that for every Metropolis-Hasting proposal $\phi^*$, we need to compute the acceptance probability $A(\phi^*; \phi)$ to move from partition $\phi$ to $\phi^*$: $A(\phi^*; \phi) = 1\, \wedge \, a(\phi^*; \phi)$, where $a(\phi^*; \phi) = \frac{\pi(\phi^*)q(\phi; \phi^*)}{\pi(\phi)q(\phi^*;\phi)} $, $q(\phi^*; \phi)$ the probability of proposing $\phi^*$ from $\phi$ and $\pi$ is the posterior distribution.

\paragraph{Sampling $\bt$} We iteratively sample the partitions $\bt_j$ for all $j$, given $\bk$ and $\bt_{-j}$. Two costumers $i_1$ and $i_2$ in restaurant $j$ are randomly picked and if they belong to the same cluster ($t_{ji_1} = t_{ji_2}$) a split move is performed, otherwise a merge move is implemented.
\begin{itemize}
\item {\bf Split} When a split move is performed, we need to sample the new table assignment of the elements in the same cluster as $i_1$ and $i_2$. This is done similarly as \cite{jain2004split}'s restricted Gibbs sampling proposal. Moreover, since a new table $t_{new}$ is created, a new dish $k_{jt_{new}}$ is sampled (uniformly among the existing dishes and a new dish). 
Note that this affects the partition of tables into dishes, so it needs to be taken into account in the likelihood. Let $\bt^*$ and $\bk^*$ represent the split proposal for the table and the dish assignments, with probability $q(\phi^* = (\bt^*,\bk^*); \phi = (\bt, \bk))$, which can be computed multiplying the probabilities of the restricted Gibbs sampling steps.
The posterior ratio $\pi(\bt^*,\bk^*)/\pi(\bt,\bk)$ is given by three main elements: the change in likelihood produced by the change in dish allocation, the change in prior probability of clustering costumers at tables, and the change in the prior for the clustering of tables into dishes. 
$$
\frac{\pi(\bt^*,\bk^*)}{\pi(\bt,\bk)} = 
\frac{p(\{y_{ji} : k^*_{jt_{ji}} = k_1 \})p(\{y_{ji} : k^*_{jt_{ji}} = k_2 \})}{p(\{y_{ji} : k_{jt_{ji}} = k_1 \})p(\{y_{ji} : k_{jt_{ji}} = k_2 \})}
\frac{\Gamma(n^*_{k_1})\Gamma(n^*_{k_2})(1+\eta\I(n_{k_2}=0))}{\Gamma(n_{k_1})\Gamma(n_{k_2})} 
 \frac{\Gamma(n^*_{t_1})\Gamma(n^*_{t_2})\alpha}{\Gamma(n_{t_1})}.
$$

\item {\bf Merge} If two tables are merged, they get assigned to the dish of $i_1$'s table and the merge happens in one unique way; however the reverse move needs to be computed. Thus, similarly to the Split move, we need to compute a {\it launch} split \citep{jain2004split} and we compute the probability to go from the launch split to the two original clusters; moreover, we compute the probability of choosing that particular dish. As before the likelihood is affected by the change in dish allocation and the prior by the change in table assignments and dish assignments.
\end{itemize}

\paragraph{Sampling $\bk$} We finally sample the partition of tables into dishes. This is similarly done using a split merge algorithm which is performed in the same way as in the DP mixture model, with the exception that now all the costumers seating at the tables corresponding to a dish are used to compute the likelihood corresponding to that cluster. Let $\bk^*$ be the proposed dish assignment that corresponds to splitting dish $k_1$ in $\bk$, with $k_2$ corresponding to a new dish in $\bk^*$. In this case the posterior ratio $\pi(\bk^*)/\pi(\bk)$ is given by 
$$
\frac{p(\{y_{ji} : k^*_{jt_{ji}} = k_1 \})p(\{y_{ji} : k^*_{jt_{ji}} = k_2 \})}{p(\{y_{ji} : k_{jt_{ji}} = k_1 \})}
\frac{\Gamma(n^*_{k_1})\Gamma(n^*_{k_2})\eta}{\Gamma(n_{k_1})}
$$

\subsection{Split-merge for nHDP}
\label{app:revJMP}

In the nested Hierarchical Dirichlet Process, restaurants are no longer fixed entities, but they are clusters of groups of costumers. Let $r_g$ be restaurant allocation of group $g$ (which is defined by the low-resolution partition), and $\br = (r_g: \forall \, g)$. Moreover let $g_c$ be the fixed mapping relating costumer $c$ to the corresponding group (i.e. which census tract contains the block group $c$) and let $r_{g_c}$ be the restaurant associated to costumer $c$ through its group $g$.
In this model, on top of sampling $\bt$ and $\bk$ given the restaurant assignment, we need to sample the partition of groups into restaurants $\br$.

We use a Metropolis-Hastings MCMC sampling algorithm, in which the chain moves from state $\phi$ to state $\phi^*$ with probability $\alpha(\phi^*, \phi) = \min\{1, A(\phi^*,\phi)\}$ and 
$$A(\phi^*,\phi) = \frac{q(\phi \vert \phi^*)p(\phi^*)}{q(\phi^* \vert \phi)p(\phi)}.$$

\paragraph{Sampling $\bt$ and $\bk$.} This step reduces to the split-merge sampling for the HDP described in the previous section, given the restaurant allocation of all the costumers defined by $(r_{g_c}: \forall c)$.
\paragraph{Sampling $\br$.} Since $\br$ defines the division of groups into restaurants, it influences the prior probability of the assignment of costumers into tables $\bt$. In fact, changing the number of costumers of a restaurant affects the probability of the partition, even when the clusters remain unchanged.
Moreover, changing $\br$ also affects the table assignment itself, because in some cases, when changing the restaurant assignment of a group, the table assignments might become incompatible with the proposed restaurant assignment. 

Consider e.g. the case of splitting a restaurant in two new restaurants, when some costumers across these two new restaurants are sitting at the same table. For two costumers to be separated in the two new restaurants, they must belong two different groups, whose restaurant assignment is changed in the split.
Specifically, let $j_1$ and $j_2$ be these two groups, such that $s = r_{j_1} = r_{j_2}$ but $s_1 = {r}^*_{j_1} \neq {r}^*_{j_2} =s_2$ and let $i_1$ and $i_2$ be two costumers such that $j_1 = g_{i_1}$ and $j_2 = g_{i_2}$. If these two costumers are sitting at the same table before the split move, $t_{si_1} = t_{si_2}$, they cannot still sit in the same table after the split move, i.e. we cannot have ${t}^*_{s_1 i_1} = {t}^*_{s_2 i_2}$, since sharing a table between two different restaurants is not possible. Thus such table assignment has probability zero given the proposed restaurant assignment and needs to be resampled together with it. That is, we need to propose, together with ${\br}^*$, a new table assignment ${\bt}^*$ such that ${t}^*_{s_1 i_1} \neq {t}^*_{s_2 i_2}$.
Moreover, as we saw in Section~\ref{app:smHDP}, to propose a new value for $\bt$, we need to also propose a new value for $\bk$. 
As a consequence, our split and merge move for $\br$ is in fact a move that affects all the assignments $\br, \bt$ and $\bk$. In other words, the chain moves from state $\phi = (\br, \bt, \bk)$ to $\phi^* = (\br^*, \bt^*, \bk^*)$, with a proposal that can be factorized conditionally:
$$
q(\phi^* \vert \phi) = q(\br^* \vert \br) q(\bt^* \vert \br^*, \bt) q(\bk^* \vert \bt^*, \bk).
$$

We randomly sample two groups indices $j_1$ and $j_2$ and if $r_{j_1} = r_{j_2} = s$ we split that restaurant; if instead $s_1 = r_{j_1} \neq r_{j_2} = s_2$ we merge the restaurants $s_1$ and $s_2$.

\begin{itemize}
\item {\bf Split} The split move needs to sample the new restaurant assignment of all the groups in restaurant $s$ except for $j_1$ and $j_2$, i.e. for $G_s = \{ j \neq j_1,j_2 : r_j = s\}$. Let $s_1 = s$ and $s_2 = K_r + 1$ be the two sub-restaurants in $\br^*$, where $K_r$ is the number of clusters in $\br$. We assign $j_1$ to $s_1$ and $j_2$ to $s_2$, that is $r^*_{j_1} = s_1$ and $r^*_{j_2} = s_2$, and we sample $r^*_j \in \{s^*_1, s^*_2\}$ for all $j \in G_s$. This is done according to $q_{\rm split}(\br^* \vert \br) = \prod_{j \in G_s} q_{\rm split}(r^*_j \vert \br^*_{-j}, \br)$, described below.

Given a proposed restaurant assignment $\br^*$, we sample the proposed table assignment $\bt^*$ and dish assignment $\bk^*$, taking into account that each new restaurant $s$ in $\br^*$ contains all the costumers $i$ whose group $g_i$ belongs to restaurant $s$, i.e. all costumers $i$ such that $r^*_{g_i} = s$. Thus the table assignment $\bt$ is changed to replace restaurant $s$ with $s_1$ and to include restaurant $s_2$. The table assignments are changed in the following way: 
\begin{itemize}
\item if the costumers sitting at table $h$ in $\bt_s$ all belong to groups that are assigned to some $s_k$ in $\br^*$, the table remains unchanged in $s_k$, for $k=1,2$. In other words, if $\{ r^*_{g_i} : t_i = h \} = \{ s_1 \}$, then table $h$ remains unchanged in $s_1$, and if  $\{ r^*_{g_i} : t_i = h \} = \{ s_2 \}$, then table $h$ remains unchanged in $s_2$. Moreover, the dish assignment of the table does not change: if $k_h = d$, then $k^*_h = d$.
\item if some of the costumers sitting at table $h$ in $\bt_s$ belong a group assigned to $s_1$ in $\br^*$, and other costumers belong to groups assigned to $s_2$ in $\br^*$, then the table is split into two sub-tables, one for each sub-restaurant. This happens if $\{ r^*_{g_i} : t_i = h \} = \{ s_1, s_2 \}$. The two subtables $h_1$ and $h_2$ are created deterministically, assigning to each one the costumers that belong to groups that are assigned to the corresponding sub-restaurant. So $t^*_i = h_k$ if and only if $r^*_{g_i} = s_k$ for $k = 1,2$ for all $i$ such that $t_i = h$. Moreover, the dish assignment of the tables does not change: if $k_h = d$, then $k^*_{h_1} = d$ and $k^*_{h_2} = d$. 
\end{itemize}
Note that these changes to the table and dish assignments, do not affect the costumer to dish assignment: even though a costumer might belong to a different restaurant or seat to a different table, its dish assignment will remain the same. 
Moreover, since the table and dish assignments are changed in a deterministic way, $q_{\rm split}(\bt^* \vert \br^*, \bt)=1$ and $q_{\rm split}(\bk^* \vert \bt^*, \bk)=1$. Thus we only need to specify $q_{\rm split}(\br^* \vert \br)$. This is done with a restricted Gibbs sampling step: $q_{\rm split}(\br^* \vert \br) = \prod_{j \in G_s} q_{\rm split}(r^*_j \vert \br^*_{-j},\br)$ and we choose 
\begin{align*}
q_{\rm split}(r^*_j = s_k \vert \br^*_{-j}, \br) &= p(r^*_j, \bt^*, \bk^* \vert \y, \br^*_{-j}) =\\
&=\frac{ p(\y \vert \bk^*, \bt^*) p(\bk^* \vert \bt^*) p(\bt^* \vert \br) n_{-j, s_k} }
{\sum_{k=1,2} p(\y \vert \bk^*, \bt^*) p(\bk^* \vert \bt^*) p(\bt^* \vert \br) n_{-j, s_k}  }, \quad k=1,2
\end{align*}
where we have denoted with $\bt^*$ and $\bk^*$ the table and dish assignments proposed in the deterministic way we just described and $n_{-j, s_k}$ is the size of cluster $s_k$ excluding element $j$. Note that since the dish assignment of each costumer remains constant, we can simplify the proposal distribution above: $p(\y \vert \bk^*, \bt^*)$ remains constant for all $\br$ considered in this restricted Gibbs sampling step. Moreover, $p(\bt^* \vert \br) \propto p(\bt^*_{s_1} \vert \br) p(\bt^*_{s_2} \vert \br)$, as the other table assignments are not affected in this step. 
Thus 
$$
q_{\rm split}(r^*_j = s_k \vert \br^*_{-j}, \br) = \frac{ p(\bk^* \vert \bt^*) p(\bt^*_{s_1} \vert \br) p(\bt^*_{s_2} \vert \br) n_{-j, s_k} }
{\sum_{k=1,2}  p(\bk^* \vert \bt^*) p(\bt^*_{s_1} \vert \br) p(\bt^*_{s_2} \vert \br)n_{-j, s_k}  }, \quad k=1,2
$$
Note that it is also thanks to these deterministic proposal distributions that by simply sampling $\br$, we can sample in the multidimensional space of $(\br, \bt, \bk)$. Additionally, note that to simplify the computations, instead of $p(\bk^* \vert \bt^*)$ we consider $p(\bk^* \vert \bt^*)/ p(\bk \vert \bt)$. This can simply be computed as $\prod_{k \in D_s} \Gamma(n_k + m_k)/\Gamma(n_k)$, where $n_k$ is the number of tables belonging to dish $k$ in the original assignment $\bk$, $m_k$ is how many of those tables were split into two sub-tables, and $D_s$ is the set of dishes served in restaurant $s$. 

\item {\bf Merge} The merge move changes the restaurant assignment of all the groups in restaurants $s_1$ and $s_2$, $G_{s_1,s_2} = \{j : r_j \in \{s_1, s_2\} \}$. Let $s$ be the new restaurant which will replace $s_1$ and let $r^*_j = s$ for all $j \in G_{s_1,s_2}$ (restaurant $s_2$ gets removed from $\br^*$). Note that $q_{\rm merge}(\br^* \vert \br) = 1$. 

As the restaurant assignment is changed, the table and dish assignments need to change too. As before, we need to take into account that a new restaurant $s$ in $\br^*$ contains all the costumers $i$ whose group $g_i$ belongs to restaurant $s$, i.e. $r^*_{g_i} = s$.
A naive proposal for changing the table assignment would be to move all tables of $s_2$ to the new merged restaurant $s$. However, we need to choose a proposal that can make the split move reversible. For this reason, in the table assignment proposal we need to merge some tables that belonged to the two restaurants. 

Specifically, for each dish $d$ we consider the tables in the two restaurants $s_1$ and $s_2$ that were assigned to dish $d$. Let $T^d_{s_k} = \{ h \; {\rm table \; in } \; s_k: k_h = d \}$ for $k=1,2$. If there is at least one such table in each restaurant, i.e. $\# T^d_{s_k}>0$ for both $k=1,2$, we combine tables into pairs. This is done by considering the restaurant with the least number of such tables, say $s_1$, and considering a one-to-one function $f$ from its tables $T^d_{s_1}$ to the ones in the other restaurant $T^d_{s_2}$, sampled uniformly at random. Thus, if $k_d = \# T^d_{s_1}$ and $n_d = \# T^d_{s_2}$, the probability of sampling $f$ is $\frac{1}{n_d!/(n_d-k_d)!}$.
Given this matching $f$, we consider the events of merging or not the tables in each pair $(h_{i,1}, h_{i,2})$ with probability $p_{h_{i,1}, h_{i,2}}$; we found the value $p_{h_{i,1}, h_{i,2}} = 0.5$ to be working well. 
Note that if some tables $h_1$ and $h_2$ are merged in table $h$, then $t^*_i = h$ for all $i$ such that $t_i \in \{h_1, h_2\}$; otherwise  $t^*_i = t_i$.
Note that in either case the dish assignment will not change, $k^*_h = d$ in the former case, or $k^*_{h_k} = d$ for $k=1,2$ in the latter.

Thus the overall probability of the new table assignment $\bt^*$ is given by
$$
q_{\rm merge}(\bt^* \vert \br^*, \bt) = \prod_{d} \left[\frac{1}{n_d!/(n_d-k_d)!}\prod_{i = 1}^{k_d} \left([p^m_{t_1, t_2}]^{\I({\rm m})}[1-p^m_{t_1, t_2}]^{\I({\rm s})} \right) \right]
$$
where $\I(m)$ and $\I(s)$ are the indicators of a split or a merge. Note that, as in the split move, the choice for $\bk^*$ is deterministic and $q_{\rm merge}(\bk^*\vert \bt^*, \bk) = 1$.
\end{itemize}

Remember now that to find the acceptance probability $\alpha(\phi^*, \phi)$ we need to consider the proposed move and the reverse move. Thus, to compute $A_{\rm split}(\phi^*,\phi)$ we have
\begin{align*}
A_{\rm split}(\phi^*,\phi) &= \frac{q_{\rm merge}(\phi \vert \phi^*) p(\phi)}{q_{\rm split}(\phi^* \vert \phi) p(\phi)} =\\
&= \frac{q_{\rm merge}(\bt \vert \br, \bt^*)}{q_{\rm split}(\br^* \vert \br) } \frac{p(\y \vert \bt, \bk) p(\bk \vert \bt) p(\bt \vert \br) p(\br)}{p(\y \vert \bt^*, \bk^*)p(\bk^* \vert \bt^*) p(\bt^* \vert \br^*) p(\br^*)}=\\
&= \frac{q_{\rm merge}(\bt \vert \br, \bt^*)}{q_{\rm split}(\br^* \vert \br) } \frac{p(\bk \vert \bt) p(\bt \vert \br) p(\br)}{p(\bk^* \vert \bt^*) p(\bt^* \vert \br^*) p(\br^*)},
\end{align*}
where the likelihood ratio $p(\y \vert \bt, \bk)/p(\y \vert \bt^*, \bk^*)$ can be ignored because it's equal to 1, as the dish assignment does not change.
Similarly, for $A_{\rm merge}(\phi^*,\phi)$ we have
\begin{align*}
A_{\rm merge}(\phi^*,\phi) &= \frac{q_{\rm split}(\phi \vert \phi^*) p(\phi)}{q_{\rm merge}(\phi^* \vert \phi) p(\phi)} =\\
&= \frac{q_{\rm split}(\br \vert \br^*)}{q_{\rm merge}(\bt^* \vert \bt, \br^*) } \frac{p(\bk \vert \bt) p(\bt \vert \br) p(\br)}{p(\bk^* \vert \bt^*) p(\bt^* \vert \br^*) p(\br^*)}.
\end{align*}

\section{Simulations}

In Section~3 of the main manuscript we present a simulation study comparing the performance of the proposed method and other competing methods. We now describe how the synthetic data was generated and present additional results.

\subsection{First framework: mixtures of normals}
\subsubsection{Synthetic data generation}

We consider six distributions, $F_k$, with $k = 1, \ldots, 6$, each of them being a mixture of up to six normals with different weights: $F_k(x) = \sum_{i = 1}^6 w_{ki} \phi((x - \mu_i)/\sigma)/\sigma$, where the means are equally distanced around zero with $\mu_i - \mu_{i-1} = 2.5$ and $\sigma = 0.5$. This level of means separation is what we consider a moderate ``cluster separation'', where the main modes are distinguishable, but there is some small mass overlap between the mixtures.
The mixture weights used $\w_k = (w_{k1}, \ldots, w_{k6})$, reported in Table \ref{tab:mixture_weights}, are chosen to make the different distributions distinguishable.

For each low resolution unit $A_\ell$, we sample $G_\ell$ uniformly among the mixtures $F_k$, $k = 1, \ldots, 6$. This is equivalent to uniformly sampling a low-resolution cluster assignment $z_\ell$, and assign $G_\ell = F_{z_\ell}$.
We then sample the high-resolution ob $\theta_{\ell, h}$ from the $(\mu_i)$ according to the weights $\w_k$ corresponding to the sampled $F_{z_\ell}$. Again, this can be seen as sampling the high-resolution cluster assignment $z_{\ell,h}$ from $i = 1, \ldots, 6$, with probabilities given by $\w_{z_\ell}$.
The observation $y_{\ell, h}$ is then sampled from a normal distribution $N(\mu_{z_{\ell,h}}, \sigma^2)$. 
In other words, the observations are sampled from $G_\ell$, i.e. from the corresponding $F_{z_\ell}$.

\begin{table}
\centering
\begin{tabular}{c c } 
$\w_1$ &= (0, 0.6, 0.3, 0, 0.1, 0) \\ 
$\w_2$ &= (0.4, 0, 0.1, 0.1, 0.4, 0) \\
$\w_3$ &= (0.1, 0, 0, 0, 0.3, 0.6) \\
$\w_4$ &= (0, 0, 0.5, 0.5, 0, 0) \\
$\w_5$ &= (0.2, 0.2, 0, 0.2, 0.2, 0.2) \\
$\w_6$ &= (0, 0, 0, 1, 0, 0) \\
\end{tabular}
\caption{Mixture weights used for each distribution $F_k$.\label{tab:mixture_weights}}
\end{table}

Note that the same low-resolution partition is used when we increase the number $n_\ell$ of high-resolution units within each low-resolution unit.

\subsubsection{Choice of hyperparameters}

We chose the prior distribution of $\sigma^2$ to be centered around $0.25$, the true value used to generate the data, with a somewhat large variance. Thus, we set $\beta_0 = 5, \beta_1 = 1$. Since the data was not standardized before fitting the model, we chose a value of $k_0 = 1/100$, to make sure that the base measure covers the whole range of the data.
In this simulation, the concentration parameters were fixed, with $\alpha_0 = 1, \alpha_1 = 0.5, \alpha_2 = 1$. The choice is motivated by a desire for a small number of clusters, together with the empirical observation that small values of $\alpha_1$ make the model more stable. 

\subsubsection{Additional results}

In Section 3 of the main paper, we presented the results for the simulation with the number of low-resolution units $L = 25$. We now present similar results for $L = 10$ and $L = 50$.

\begin{figure}[H]
\centering
\includegraphics[width = \textwidth]{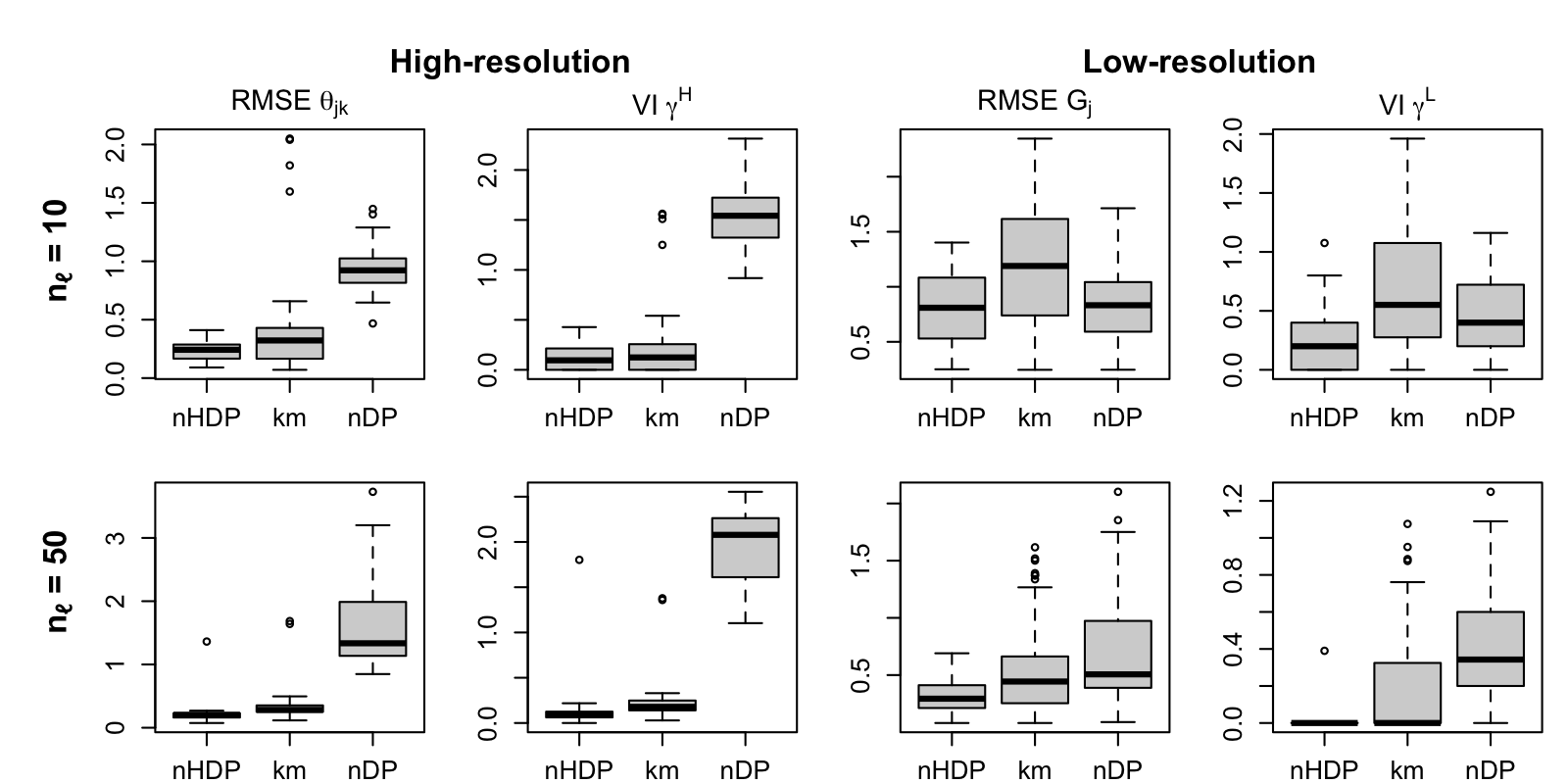} 
\caption{Parameter estimation and partition recovery measures at high-resolution (left panels) and low-resolution (right panels) for synthetic data with $L = 10$ and $n_\ell = 10$ (top panels) and $n_\ell = 50$ (bottom panels). \label{fig:sim_L10}}
\end{figure}

Figure~\ref{fig:sim_L10} and ~\ref{fig:sim_L50} show the measures of parameter estimation and partition recovery for both the high and the low-resolution levels, respectively when $L = 10$ and $L = 50$. 
We note that the methods have similar performances to the case of $L = 25$ reported in the main manuscript. In particular, all the methods tend to have worst low-resolution estimation and partition recovery when the number of low resolution units is smaller ($n_\ell = 10$). This is expected, as the estimation of the $G_\ell$ relies on the number of units within each low-resolution unit $A_\ell$.
The adapted version of k-means has particularly poor performance when $n_\ell$ is small, for both values of $L$.
Instead, when the number $n_\ell$ of high-resolution units increases, the k-means method performs as well as the nHDP. 
Unfortunately, the nDP method tends to have poorer estimation and recovery performance both for the high and low-resolution levels, and for all values of $n_\ell$ and $L$.

\begin{figure}[H]{}
\centering
\includegraphics[width = \textwidth]{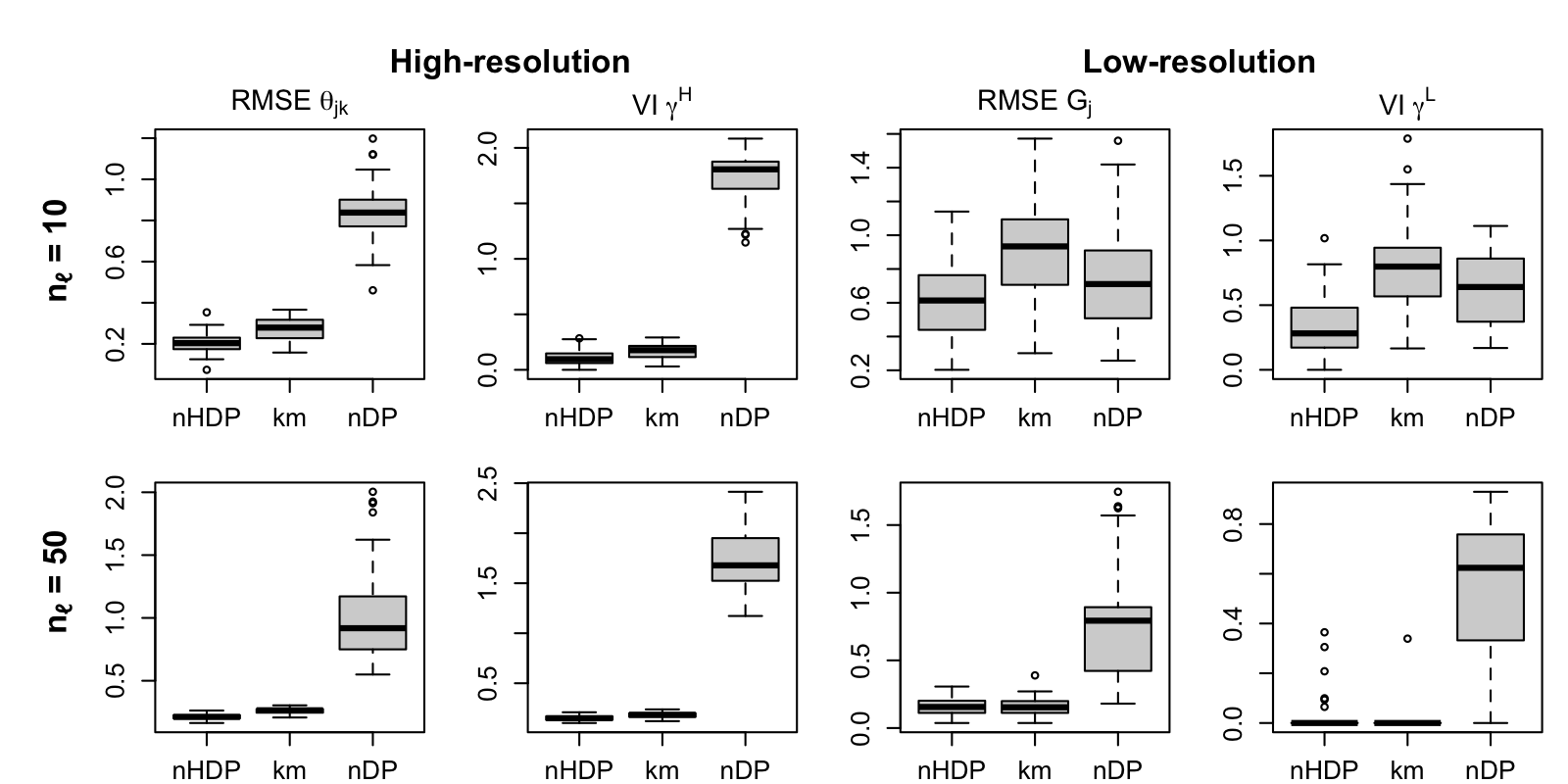} 
\caption{Parameter estimation and partition recovery measures at high-resolution (left panels) and low-resolution (right panels) for synthetic data with $L = 50$ and $n_\ell = 10$ (top panels) and $n_\ell = 50$ (bottom panels). \label{fig:sim_L50}}
\end{figure}

\subsection{Second framework: data generated from the model}
\subsubsection{Synthetic data generation}

To generate simulated datasets with distinguishable low resolution clusters, we implemented the following steps:
\begin{enumerate}
\item Using the stick-breaking construction for the Dirichlet Process, we approximate $Q = \sum_i w_i \delta_{G^*_i}$ with a truncated $\tilde{Q} = \sum_{i=1}^m \tilde{w}_i \delta_{\tilde{G}^*_i}$, where $m$ is the number of low-resolution units. We generate $\tilde{w}_1, \ldots, \tilde{w}_{m-1}$ and compute $\tilde{w}_{m} = 1-\sum_{i=1}^{m-1} \tilde{w}_i$.
\item We generate the cluster assignments of the low-resolution units $z_j$ by sampling them from $\sum_{i=1}^m \tilde{w_i} \delta_i$.
\item We generate the atoms $\tilde{G}^*_i$ of $\tilde{Q}$ drawing from a finite-dimensional approximation of the HDP, with the stick-breaking constructions of \cite{teh2006hierarchical} truncated at the number of high-resolution units $n$. 
Note that to make sure the different low-resolution clusters are distinguishable, it is important that the $\tilde{G}^*_i$ are sufficiently different.
Thus we created a rejection sampling algorithm that generates a vector of cluster-specific discrete distributions $\tilde{G}^*_1, \ldots, \tilde{G}^*_m$ such that for any $i,j$, the total variation distance $TV(\tilde{G}^*_i, \tilde{G}^*_j) > \epsilon$. We found that $\epsilon = 0.8$ was generating distributions that were sufficiently different and could be distinguished for relatively low sample sizes. 
The rejection algorithm is as following: given the first $\tilde{G}^*_1, \ldots, \tilde{G}^*_{i-1}$, generate a new distribution $\tilde{G}^*$; if the TV distance between $\tilde{G}^*$ and anyone of the previous $\tilde{G}^*_j$ with $j < i$ is less or equal than $\epsilon$, then discard it and sample a new $\tilde{G}^*$, repeat the test until a suitable $\tilde{G}^*$ is sampled and set $\tilde{G}^*_i$ equal to it.
\item Given the set of distinct $\tilde{G}^*_i$, we sample the cluster assignments of the high-resolution units $z_{lj}$ from $\sum_{k=1}^n \tilde{p}_{z_j k} \delta_k$ where $\tilde{p}_{z_j k}$ are the weights of $\tilde{G}^*_{z_j}$, the discrete measure associated to the low-resolution cluster of unit $j$.
\item Given the high-resolution clusters, we generate the cluster-specific means $\theta^*_k$. Instead of drawing them i.i.d. from a base distribution, we again make sure that the values chosen are distinct, so that the clusters are actually distinguishable. We choose them evenly spaced, centered at 0, with minimum distance $\theta^*_{k+1}-\theta^*_{k} = \kappa \sigma$, with $\kappa = 5$ in the medium cluster separation framework, and $\kappa = 8$ in the high cluster separation framework. 
\item The high-resolution data is then generated, with $y_{lj} \sim N(\theta^*_{z_{lj}}, \sigma^2)$, with a fixed value of $\sigma = 0.5$. For out-of-sample evaluation, we also generate a second set of data, $y^*_{lj} \sim N(\theta^*_{z_{lj}}, \sigma^2)$.
\item When used, the low-resolution data is aggregated by averaging the high-resolution data: $y_j = \frac{1}{m}\sum_{l=1}^m y_{jl}$ (equivalently $y^*_j = \frac{1}{m}\sum_{l=1}^m y^*_{jl}$).
\end{enumerate}

Note that generating a realization of the nHDP (step 1 and step 3) depends on the parameters $\alpha_0, \alpha_1, \alpha_2$. We consider two configurations, where $(\alpha_0, \alpha_1, \alpha_2) = (1,1,1)$ and $(\alpha_0, \alpha_1, \alpha_2) = (5,3,3)$. These two configurations allow us to generate different kind of data. In the first, the different $G^*_j$'s are quite similar, or have most support on different points, and we expect the nDP to perform quite well under this framework. In the second configuration instead, the variation among the $G^*_j$'s is larger, while having very similar support.

\begin{figure}[h]{}
\centering
\includegraphics[width = \textwidth]{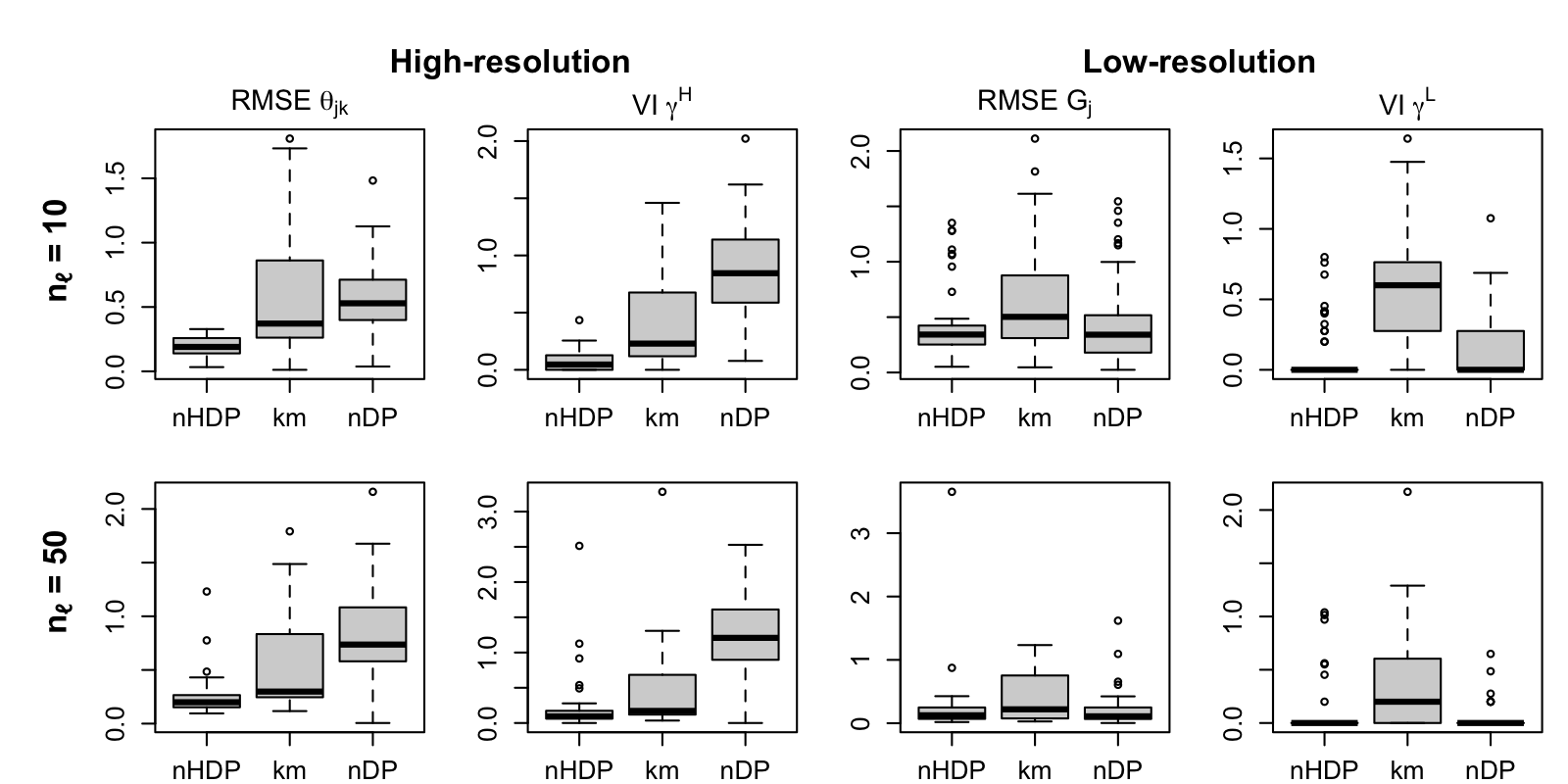}
\caption{Parameter estimation and partition recovery measures at high-resolution (left panels) and low-resolution (right panels) for synthetic data generated from the model, with $L = 10$. The top panels show the results for $n_\ell = 10$, while the bottom panels show those for $n_\ell = 50$, under the configuration with $\alpha_0 = 1, \alpha_1 = 1, \alpha_2 = 1$. \label{fig:sim_111}}
\end{figure}

\subsubsection{Choice of hyperparameters}

Similarly to the first framework, we chose the prior distribution of $\sigma^2$ to be centered around $0.25$, the true value used to generate the data, with a somewhat large variance. Thus, we set $\beta_0 = 5, \beta_1 = 1$. Since the data was not standardized before fitting the model, we chose a value of $k_0 = 1/100$, to make sure that the base measure covers the whole range of the data.
In this simulation framework, the concentration parameters were fixed equal to the true values used to generate the data.

\subsubsection{Results}

Similarly to the results reported in the main manuscript, when the synthetic data is generated from the model, we find that the nHDP achieves better performance compared to the nDP and k-means, for both parameter estimation and partition recovery at each resolution. Under the configuration given by the hyper-parameters $(\alpha_0, \alpha_1, \alpha_2) = (1,1,1)$ (Figure~\ref{fig:sim_111}), we see that the performance of the nDP is quite good, especially for the low-resolution measures. This is particularly evident in the case with $n_\ell = 50$, where the low-resolution partition recovered by both nDP and nHDP is almost always the true one; however, the nDP has much worse high-resolution estimation and partition recovery compared to the nHDP.

\begin{figure}[h]{}
\centering
\includegraphics[width = \textwidth]{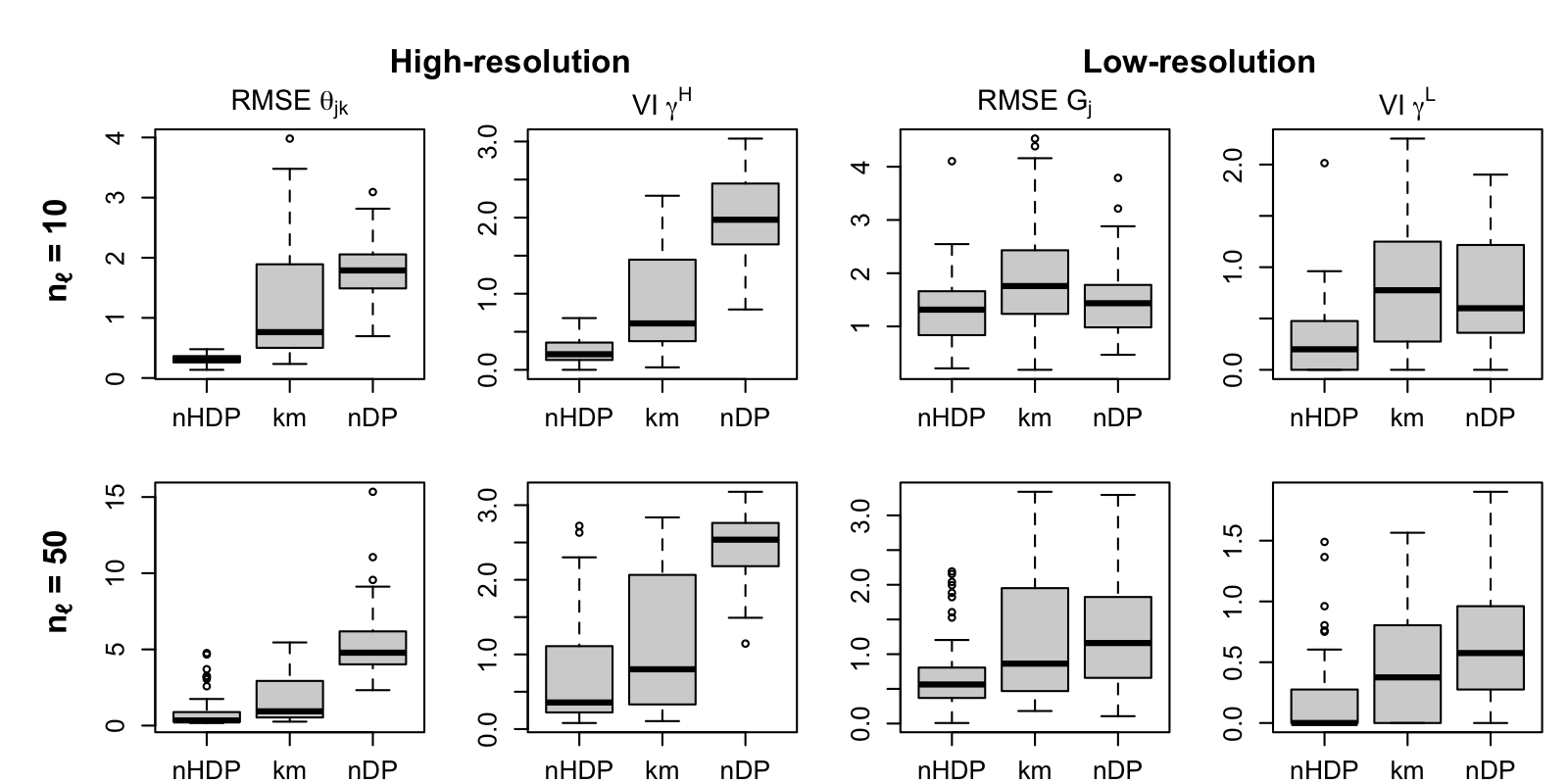} 
\caption{Parameter estimation and partition recovery measures at high-resolution (left panels) and low-resolution (right panels) for synthetic data generated from the model, with $L = 10$. The top panels show the results for $n_\ell = 10$, while the bottom panels show those for $n_\ell = 50$, under the configuration with $\alpha_0 = 5, \alpha_1 = 3, \alpha_2 = 3$. \label{fig:sim_353}}
\end{figure}

In the configuration where the distributions have stronger support overlap and variation, generated with the hyper-parameters $(\alpha_0, \alpha_1, \alpha_2) = (5,3,3)$ (Figure~\ref{fig:sim_353}), the nDP performs quite poorly, in both situations where $n_\ell = 10$ and $n_\ell = 50$, for both the measures related to the high-resolution and the low-resolution partition.

In both cases, k-means shows quite poor performance, for both partitions and parameter estimations. Contrary to the simulation with mixtures of normals, k-means did not show a significant improvement of performance for larger values of $n_\ell$.


\subsection{Details on parameter estimation measures}

In our simulations, we report a measure of parameter estimation and one of partition recovery for each resolution.
As measure of parameter estimation for $\gamma^{\mathcal{H}}$, we report the root mean squared error (RSME) for the estimation of the mean of the data $\theta_{\ell,h} = \E[y_{\ell,h}]$. For $\gamma^{\mathcal{L}}$, we use the RMSE for the estimation of the means of the mixtures assigned to the $\ell$-th low resolution unit $\phi_\ell = \E[\theta_{\ell,h} \vert G_\ell = F_j] = \sum_{k=1}^6 w_{j,k} \mu_{k}$. 
As estimators, for the data means we use the posterior mean marginally on the high resolution partition, $\hat{\theta}_{\ell,h} = \hat{\E}[\theta_{\ell,h} \vert \y]$; for the distributional means we use the estimator given by $\hat{\phi}_{\ell} = \hat{\E}[ \hat{\E}[\phi_\ell \vert \gamma^{\mathcal{L}}] \vert \y]$, where if the $\ell$th low-resolution unit is in the $k$th cluster $\mathcal{C}_k$ (i.e. if $G_\ell = F_k$) then $\hat{\E}[\phi_\ell \vert \gamma^{\mathcal{L}}] = \frac{1}{\vert \mathcal{C}_k \vert}\sum_{\ell \in \mathcal{C}_k}  (\sum_h y_{\ell,h})/n_\ell $. In other words, the distributional mean for each cluster is found by averaging the low-resolution data (i.e. high-resolution data which has been aggregated at the low-resolution) for the units in that (low resolution) cluster.

\section{Analysis of crime density in West Philadelphia}

\subsection{Choice of hyperparameters}

In Section~2.3 of the main manuscript, we presented the full model we use to describe crime density in Philadelphia. We now describe how we chose the hyperparameters $\beta_0, \beta_1$ that specify the prior for $\sigma^2,$ and the hyperparameter $k_0$ which determines the ratio of within-cluster and between-cluster variation. 

Before fitting our model to the data we standardized the $y_{\ell, h}$ to have overall mean 0 and standard deviation 1. We expect the within-cluster standard deviation of the standardized data to be far less than 1; in particular a value less than 0.5 (i.e. less than half of the overall data standard deviation).
We then chose the distribution of $\sigma^2$ to have a prior mean of 0.25, and a prior variance of 0.1. This prior assigns approximately 0.7 probability to a value of $\sigma^2$ less than 1, and less than 0.02 probability to a value of $\sigma^2$ greater than 1. 
We also expect the within-cluster standard deviation not to be too small, as that would cause the data to be split into a very large number of clusters. We chose a value of $k_0$ that would give a large prior probability that $\sigma^2/k_0$ is greater than 1 (the variance of the standardized data), given the chosen prior distribution for $\sigma^2$; in other words, we want the base distribution $H_0$ to cover the overall distribution of the data. We chose $k_0 = 1/10$, which ensures a prior probability of more than 0.8 that $H_0$ covers the data.

We finally specify a prior distribution for the concentration parameters $\alpha_0, \alpha_1, \alpha_2$. 
Remember that under the DP model with concentration parameter $\eta$, the average asymptotic number of clusters for $n$ observations is $\eta \log(n)$. To avoid a too small or too large number of clusters, we chose to model the prior as a truncated normal with prior mean equal to 2 and prior standard deviation 1.

\subsection{Partition estimation}

Following the recommendation of \cite{wade2018bayesian}, we summarize the retained posterior samples finding the partition that minimizes the Variation of Information (VI) distance from the posterior samples. We used the package \cite{mcclustext} to find the best candidate among the MCMC draws. 
We additionally searched for the best candidate among the partitions found with hierarchical clustering (the function \texttt{hclust} in R), using different number of clusters, from 1 to the maximum number of clusters found by the MCMC. This essentially equivalent to using the function \texttt{mcclust.ext::minVI} with the option \texttt{method = "all"}, with the difference that we are choosing the partition that minimizes the actual VI distance (using function \texttt{mcclust.ext::VI}) instead of its lower bound (\texttt{mcclust.ext::VI.lb}).
The partitions reported in the analysis presented in the main manuscript are the ones that achieved the minimum VI distance.

\section{Graphical representation of models}

In Figure~\ref{fig:diagram} we represent pictorially a possible realization of the discrete measures in the nDP and HDP models, and how they are combined to construct the nHDP: as the diagram shows, both the nDP and the nHDP share the discrete measure $Q$ (represented by the large rectangle), but its atoms $Q^*_j$ differ. In the nDP, all the measures $Q^*_j$ have different atoms locations, which is the reason why the nDP produces only partitions whose clusters are nested within clusters. In fact, if two $G_\ell$ are not equal to the same $Q^*_j$, they will not have the same support and the corresponding $\theta_{\ell,h}$ will not have a chance to be equal.
In the nHDP instead, the the $Q^*_j$ are equivalent to the $G_\ell$ in the HDP, i.e. different $Q^*_j$'s have the same atoms but different weights (the vertical lines are located in the same locations but have different height). 
While in the HDP each represented discrete distribution is a group-specific distribution $G_\ell$ and all of them are different because they have different weights, the nHDP allows some of the group-specific measures $G_\ell$ to be equal to the same $Q^*_j$, thanks to the discrete measure $Q$. This allows clustering of the low-resolution units, contrary to the HDP. Note that in the nHDP, even when two groups have different $Q^*_j$'s, they share the same atoms $\theta_k$, which allows the parameters of the high resolution units $\theta_{\ell,h}$ to be clustered together even if their groups $G_{\ell_1}$ and $G_{\ell_2}$ are not clustered together.

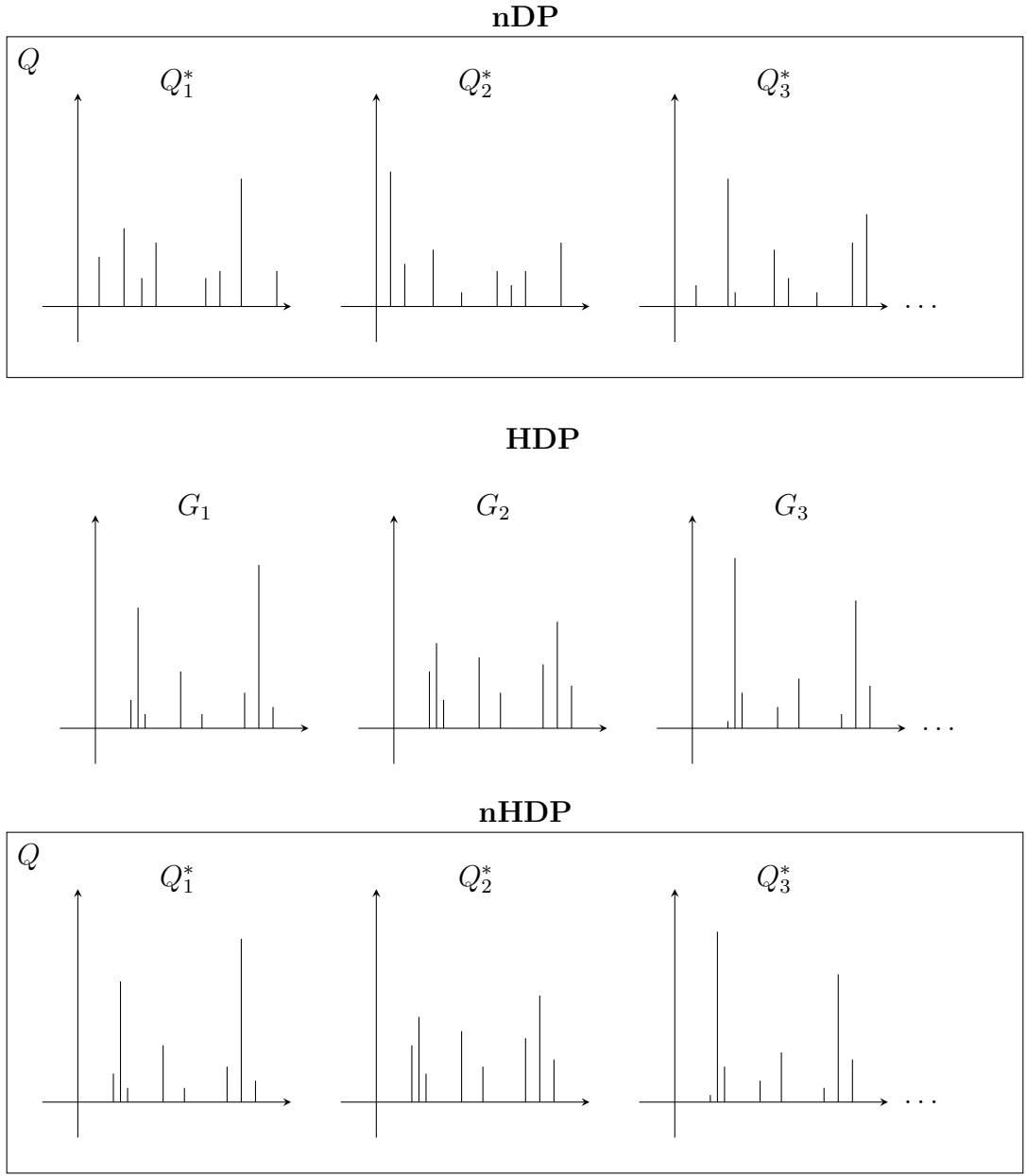
\begin{figure}[H]
\centering
\begin{tikzpicture}

\pgfmathsetmacro\len{2.8}     
\pgfmathsetmacro\spa{1.4}
\pgfmathsetmacro\xspa{1.4}   

\coordinate (a) at (0,0);
\coordinate (b) at (\len+\spa,0);
\coordinate (c) at (2*\len+2*\spa,0);

\node[right] at ($(a) + (-1,\len+0.65)$) {$Q$};
\node[above] at (3*\len/2+\spa+\spa/2,\len+1) {\textbf{nDP}};

\draw ($(a) - (1,1)$) rectangle ($(c) + (\len+\spa/2+\xspa,\len+1)$);

\draw [-stealth]($(a) + (-0.5,0)$) -- ($(a) + (3,0)$);
\draw [-stealth]($(a) + (0,-0.5)$) -- ($(a) + (0,3)$);
\draw ($(a) + (0.3,0)$) -- ($(a) + (0.3,0.7)$);
\draw ($(a) + (0.65,0)$) -- ($(a) + (0.65,1.1)$);
\draw ($(a) + (0.9,0)$) -- ($(a) + (0.9,0.4)$);
\draw ($(a) + (1.1,0)$) -- ($(a) + (1.1,0.9)$);
\draw ($(a) + (1.8,0)$) -- ($(a) + (1.8,0.4)$);
\draw ($(a) + (2.0,0)$) -- ($(a) + (2.0,0.5)$);
\draw ($(a) + (2.3,0)$) -- ($(a) + (2.3,1.8)$);
\draw ($(a) + (2.8,0)$) -- ($(a) + (2.8,0.5)$);
\node[above] at ($(a) + (\len/2,\len)$){$Q_1^*$};

\draw [-stealth]($(b) + (-0.5,0)$) -- ($(b) + (3,0)$);
\draw [-stealth]($(b) + (0,-0.5)$) -- ($(b) + (0,3)$);
\draw ($(b) + (0.2,0)$) -- ($(b) + (0.2,1.9)$);
\draw ($(b) + (0.4,0)$) -- ($(b) + (0.4,0.6)$);
\draw ($(b) + (0.8,0)$) -- ($(b) + (0.8,0.8)$);
\draw ($(b) + (1.2,0)$) -- ($(b) + (1.2,0.2)$);
\draw ($(b) + (1.7,0)$) -- ($(b) + (1.7,0.5)$);
\draw ($(b) + (1.9,0)$) -- ($(b) + (1.9,0.3)$);
\draw ($(b) + (2.1,0)$) -- ($(b) + (2.1,0.5)$);
\draw ($(b) + (2.6,0)$) -- ($(b) + (2.6,0.9)$);
\node[above] at ($(b) + (\len/2,\len)$){$Q_2^*$};

\draw [-stealth]($(c) + (-0.5,0)$) -- ($(c) + (3,0)$);
\draw [-stealth]($(c) + (0,-0.5)$) -- ($(c) + (0,3)$);
\draw ($(c) + (0.3,0)$) -- ($(c) + (0.3,0.3)$);
\draw ($(c) + (0.75,0)$) -- ($(c) + (0.75,1.8)$);
\draw ($(c) + (0.85,0)$) -- ($(c) + (0.85,0.2)$);
\draw ($(c) + (1.4,0)$) -- ($(c) + (1.4,0.8)$);
\draw ($(c) + (1.6,0)$) -- ($(c) + (1.6,0.4)$);
\draw ($(c) + (2.0,0)$) -- ($(c) + (2.0,0.2)$);
\draw ($(c) + (2.5,0)$) -- ($(c) + (2.5,0.9)$);
\draw ($(c) + (2.7,0)$) -- ($(c) + (2.7,1.3)$);
\node[above] at ($(c) + (\len/2,\len)$){$Q_3^*$};

\node at ($(c) + (\len+\spa/2,0)$) {\ldots};

\end{tikzpicture} \\

\vspace{15pt}

\begin{tikzpicture} 

\pgfmathsetmacro\len{2.8}     
\pgfmathsetmacro\spa{1.4}
\pgfmathsetmacro\xspa{1.4}  

\coordinate (a) at (0,0);
\coordinate (b) at (\len+\spa,0);
\coordinate (c) at (2*\len+2*\spa,0);

\node[above] at (3*\len/2+\spa+\spa/2,\len+1) {\textbf{HDP}};


\draw [-stealth]($(a) + (-0.5,0)$) -- ($(a) + (3,0)$);
\draw [-stealth]($(a) + (0,-0.5)$) -- ($(a) + (0,3)$);
\draw ($(a) + (0.5,0)$) -- ($(a) + (0.5,0.4)$);
\draw ($(a) + (0.6,0)$) -- ($(a) + (0.6,1.7)$);
\draw ($(a) + (0.7,0)$) -- ($(a) + (0.7,0.2)$);
\draw ($(a) + (1.2,0)$) -- ($(a) + (1.2,0.8)$);
\draw ($(a) + (1.5,0)$) -- ($(a) + (1.5,0.2)$);
\draw ($(a) + (2.1,0)$) -- ($(a) + (2.1,0.5)$);
\draw ($(a) + (2.3,0)$) -- ($(a) + (2.3,2.3)$);
\draw ($(a) + (2.5,0)$) -- ($(a) + (2.5,0.3)$);
\node[above] at ($(a) + (\len/2,\len)$){$G_1$};

\draw [-stealth]($(b) + (-0.5,0)$) -- ($(b) + (3,0)$);
\draw [-stealth]($(b) + (0,-0.5)$) -- ($(b) + (0,3)$);
\draw ($(b) + (0.5,0)$) -- ($(b) + (0.5,0.8)$);
\draw ($(b) + (0.6,0)$) -- ($(b) + (0.6,1.2)$);
\draw ($(b) + (0.7,0)$) -- ($(b) + (0.7,0.4)$);
\draw ($(b) + (1.2,0)$) -- ($(b) + (1.2,1.0)$);
\draw ($(b) + (1.5,0)$) -- ($(b) + (1.5,0.5)$);
\draw ($(b) + (2.1,0)$) -- ($(b) + (2.1,0.9)$);
\draw ($(b) + (2.3,0)$) -- ($(b) + (2.3,1.5)$);
\draw ($(b) + (2.5,0)$) -- ($(b) + (2.5,0.6)$);
\node[above] at ($(b) + (\len/2,\len)$){$G_2$};

\draw [-stealth]($(c) + (-0.5,0)$) -- ($(c) + (3,0)$);
\draw [-stealth]($(c) + (0,-0.5)$) -- ($(c) + (0,3)$);
\draw ($(c) + (0.5,0)$) -- ($(c) + (0.5,0.1)$);
\draw ($(c) + (0.6,0)$) -- ($(c) + (0.6,2.4)$);
\draw ($(c) + (0.7,0)$) -- ($(c) + (0.7,0.5)$);
\draw ($(c) + (1.2,0)$) -- ($(c) + (1.2,0.3)$);
\draw ($(c) + (1.5,0)$) -- ($(c) + (1.5,0.7)$);
\draw ($(c) + (2.1,0)$) -- ($(c) + (2.1,0.2)$);
\draw ($(c) + (2.3,0)$) -- ($(c) + (2.3,1.8)$);
\draw ($(c) + (2.5,0)$) -- ($(c) + (2.5,0.6)$);
\node[above] at ($(c) + (\len/2,\len)$){$G_3$};

\node at ($(c) + (\len+\spa/2,0)$) {\ldots};
\end{tikzpicture} \\

\vspace{10pt}

\begin{tikzpicture}

\pgfmathsetmacro\len{2.8}     
\pgfmathsetmacro\spa{1.4}
\pgfmathsetmacro\xspa{1.4}

\coordinate (a) at (0,0);
\coordinate (b) at (\len+\spa,0);
\coordinate (c) at (2*\len+2*\spa,0);

\node[right] at ($(a) + (-1,\len+0.65)$) {$Q$};
\node[above] at (3*\len/2+\spa+\spa/2,\len+1) {\textbf{nHDP}};

\draw ($(a) - (1,1)$) rectangle ($(c) + (\len+\spa/2+\xspa,\len+1)$);

\draw [-stealth]($(a) + (-0.5,0)$) -- ($(a) + (3,0)$);
\draw [-stealth]($(a) + (0,-0.5)$) -- ($(a) + (0,3)$);
\draw ($(a) + (0.5,0)$) -- ($(a) + (0.5,0.4)$);
\draw ($(a) + (0.6,0)$) -- ($(a) + (0.6,1.7)$);
\draw ($(a) + (0.7,0)$) -- ($(a) + (0.7,0.2)$);
\draw ($(a) + (1.2,0)$) -- ($(a) + (1.2,0.8)$);
\draw ($(a) + (1.5,0)$) -- ($(a) + (1.5,0.2)$);
\draw ($(a) + (2.1,0)$) -- ($(a) + (2.1,0.5)$);
\draw ($(a) + (2.3,0)$) -- ($(a) + (2.3,2.3)$);
\draw ($(a) + (2.5,0)$) -- ($(a) + (2.5,0.3)$);
\node[above] at ($(a) + (\len/2,\len)$){$Q_1^*$};

\draw [-stealth]($(b) + (-0.5,0)$) -- ($(b) + (3,0)$);
\draw [-stealth]($(b) + (0,-0.5)$) -- ($(b) + (0,3)$);
\draw ($(b) + (0.5,0)$) -- ($(b) + (0.5,0.8)$);
\draw ($(b) + (0.6,0)$) -- ($(b) + (0.6,1.2)$);
\draw ($(b) + (0.7,0)$) -- ($(b) + (0.7,0.4)$);
\draw ($(b) + (1.2,0)$) -- ($(b) + (1.2,1.0)$);
\draw ($(b) + (1.5,0)$) -- ($(b) + (1.5,0.5)$);
\draw ($(b) + (2.1,0)$) -- ($(b) + (2.1,0.9)$);
\draw ($(b) + (2.3,0)$) -- ($(b) + (2.3,1.5)$);
\draw ($(b) + (2.5,0)$) -- ($(b) + (2.5,0.6)$);
\node[above] at ($(b) + (\len/2,\len)$){$Q_2^*$};

\draw [-stealth]($(c) + (-0.5,0)$) -- ($(c) + (3,0)$);
\draw [-stealth]($(c) + (0,-0.5)$) -- ($(c) + (0,3)$);
\draw ($(c) + (0.5,0)$) -- ($(c) + (0.5,0.1)$);
\draw ($(c) + (0.6,0)$) -- ($(c) + (0.6,2.4)$);
\draw ($(c) + (0.7,0)$) -- ($(c) + (0.7,0.5)$);
\draw ($(c) + (1.2,0)$) -- ($(c) + (1.2,0.3)$);
\draw ($(c) + (1.5,0)$) -- ($(c) + (1.5,0.7)$);
\draw ($(c) + (2.1,0)$) -- ($(c) + (2.1,0.2)$);
\draw ($(c) + (2.3,0)$) -- ($(c) + (2.3,1.8)$);
\draw ($(c) + (2.5,0)$) -- ($(c) + (2.5,0.6)$);
\node[above] at ($(c) + (\len/2,\len)$){$Q_3^*$};

\node at ($(c) + (\len+\spa/2,0)$) {\ldots};
\end{tikzpicture}

\caption[Diagram of the nHDP model.]{Possible realization of the discrete measures in the HPD, nDP and nHDP models. For the nDP and the nHDP models, the discrete measure $Q$ is represented as a rectangular box containing other discrete distributions $Q^*_j$ as its atoms. For the HDP instead, only the group-specific $G_\ell$ are represented.
The graphical depiction of each discrete distribution uses vertical lines to represent the atoms of the distribution: the location of each line represents the location of the atom, and the height of each line represents the atom's weight or probability. For the purpose of the plot, only a finite number of atoms are depicted. {\bf Top panel:} Diagram of the nDP model. {\bf Middle panel:} Diagram of the HDP model. {\bf Bottom panel:} Diagram of the nHDP model. 
}
\label{fig:diagram}
\end{figure}

\end{document}